\documentclass[a4paper,11pt]{article}
\usepackage{amsmath,amsfonts,amssymb,amsthm}
\usepackage{latexsym}
\usepackage{color}
\usepackage{graphicx} 

\usepackage{hyperref}
\usepackage{enumerate}

\usepackage[margin=2.5cm]{geometry}

\newtheorem{definition}{Definition}[section]
\newtheorem{lemma}[definition]{Lemma}
\newtheorem{proposition}[definition]{Proposition}
\newtheorem{theorem}[definition]{Theorem}

\definecolor{detailsgray}{gray}{0.3}

\def \cbb{\mathbb{C}}

\def \nbb{\mathbb{N}}
\def \rbb{\mathbb{R}}

\def \ccal {\mathcal{C}}
\def \dcal {\mathcal{D}}

\def \fcal {\mathcal{F}}

\def \hcal {\mathcal{H}}

\def \qcal {\mathcal{Q}}

\def \scal {\mathcal{S}}

\def \bfk  {\mathfrak{B}}

\newcommand{\thetav}{\boldsymbol{\theta}}
\newcommand{\etav}{\boldsymbol{\eta}}
\newcommand{\lambdav}{\boldsymbol{\lambda}}

\newcommand{\xiv}{\boldsymbol{\xi}}

\def \. { \,\! }

\def\clap#1{\hbox to 0pt{\hss#1\hss}}

\def\mathclap{\mathpalette\mathclapinternal}

\def\mathclapinternal#1#2{\clap{$\mathsurround=0pt#1{#2}$}}

\def \cdotarg { \, \cdot \, }

\newcommand{\id}{\operatorname{id}}
\newcommand{\idop}{\boldsymbol{1}}

\def \st {^\ast}

\def \restrict {\lceil}

\def \expltext#1 {\\ \text{\footnotesize{ (#1) }}\\}
\def \intercomm#1 {\\ \text{\footnotesize{ (#1) }}\\}
\def \undercomm#1 {\underset{\text{\scriptsize{ (#1) }}}}
\def \overcomm#1 {\overset{\text{\scriptsize{ (#1) }}}}

\newcommand{\rhs}[1]{\mathrm{r.h.s.}\eqref{#1}}

\newcommand{\Hil}{\mathcal{H}}

\newcommand{\fpn}{\Hil^{\mathrm{f}}}
\newcommand{\fpno}{\Hil^{\omega,\mathrm{f}}}
\newcommand{\fpnp}{P^\mathrm{f}}

\newcommand{\boundedops}{\bfk(\Hil)}

\newcommand{\qf}{\mathcal{Q}}

\newcommand{\hscalar}[2]{\langle #1 , #2 \rangle }
\newcommand{\bighscalar}[2]{\big\langle  #1 , #2 \big\rangle }

\newcommand{\gnorm}[2]{\lVert #1 \rVert_{#2}}

\newcommand{\onorm}[2]{\lVert #1 \rVert_{#2}^\omega}

\newcommand{\ad}{a^\dagger}
\newcommand{\zd}{z^\dagger}

\newcommand{\lvector}[2]{\boldsymbol{\ell}_{#1}(#2)}
\newcommand{\rvector}[2]{\boldsymbol{r}_{#1}(#2)}

\newcommand{\cmeA}[2]{f_{#1}^{\lbrack #2 \rbrack}}
\newcommand{\cmeB}[2]{f_{#1}{\lbrack #2 \rbrack}}
\newcommand{\cme}[2]{\mathchoice{\cmeA{#1}{#2}}{\cmeB{#1}{#2}}{\cmeB{#1}{#2}}{\cmeB{#1}{#2}}}
\newcommand{\cmelong}[2]{f_{#1} \left\lbrack #2 \right\rbrack}

\newcommand{\perms}[1]{\mathfrak{S}_{#1}}

\numberwithin{equation}{section}

\title{An operator expansion for integrable quantum field theories}
\author{Henning Bostelmann\thanks{%
University of York, Department of Mathematics, York YO10 5DD, United Kindom. 
E-mail: \href{mailto:henning.bostelmann@york.ac.uk}{\nolinkurl{henning.bostelmann@york.ac.uk}}
}
 \and Daniela Cadamuro\thanks{%
 Courant Research Centre ``Higher Order Structures'', University of G\"ottingen, Germany. 
E-mail: \href{daniela.cadamuro@theorie.physik.uni-goettingen.de}{\nolinkurl{daniela.cadamuro@theorie.physik.uni-goettingen.de}}
}
}

\date{November 20, 2012}

\begin{document}

\maketitle

\begin{abstract}

A large class of quantum field theories on 1+1 dimensional Minkowski space, namely, certain integrable models, has recently been constructed rigorously by Lechner. However, the construction is very abstract and the concrete form of local observables in these models remains largely unknown. Aiming for more insight into their structure, we establish a series expansion for observables, similar but not identical to the well-known form factor expansion. This expansion will be the basis for a characterization and explicit construction of local observables, to be discussed elsewhere. Here, we establish the expansion independent of the localization aspect, and analyze its behavior under space-time symmetries. We also clarify relations with deformation methods in quantum field theory, specifically, with the warped convolution in the sense of Buchholz and Summers.
\end{abstract}


\section{Introduction}

The central concept of relativistic quantum physics is the notion of \emph{local observables} \cite{Haa:LQP}. These are operators associated with space-time points or, more generally, bounded space-time regions, so that operators associated with spacelike separated regions commute. This concept is at the heart of the physical interpretation of theories; for example, it enables one to identify scattering states \cite{BuchholzSummers:2006} and analyze the charge structure of the system \cite{Doplicher:2010}.

It is an unfortunate fact that local observables are extremely hard to construct in models beyond the simplistic situation of free particles. With \emph{construction}, we refer here to a rigorous proof of existence in any mathematical framework, for example as Wightman distributions \cite{StrWig:PCT}, algebras of bounded operators \cite{Haa:LQP}, or closed unbounded operators affiliated with these; techniques are available for passing between these formulations \cite{BorYng:positivity,FreHer:pointlike_fields,RehWol:fields_localized,Bos:short_distance_structure}. We are not thinking of formal perturbation theory though. This construction problem for interacting models is still beyond reach in physical space-time dimensions. Progress has been made in simplified lower-dimensional models, notably through the results of Glimm and Jaffe \cite{GliJaf:quantum_physics}; but even in these models, local observables need to be constructed in a very intricate and, in places, rather indirect way.

Our interest here is in a particular class of models on 1+1-dimensional Minkowski space, namely those with a factorizing scattering matrix (also often called \emph{integrable} due to their classical counterparts). While these models are often considered in a thermodynamical context, we view them here as relativistic quantum field theories. 

The traditional approach used for the construction of local observables in these models is the form factor programme \cite{Smirnov:1992,BabujianFoersterKarowski:2006}. Here one expands expectation values of operators into a series of form factors (asymptotic scattering states). The form factors of local observables, which in this case take the form of pointlike localized quantum fields, are subject to an infinite set of constraints. Starting from these, the form factors of local fields have been computed explicitly in various models. The remaining problem is then to control the convergence of the (necessarily infinite) sum of form factors, in order to construct Wightman $n$-point functions. Despite some recent progress \cite{BabujianKarowski:2004}, the convergence of this series remains an open problem.

A radically different approach was proposed by Schroer \cite{Schroer:1999} and carried further by Lechner \cite{Lechner:2003,Lechner:2008}. Instead of dealing with local observables in bounded regions directly, they first construct observables associated with certain unbounded regions, specifically, wedges extending to spacelike infinity. Due to fewer constraints from localization, these operators are much easier to handle explicitly. Bounded regions, say, doubles cones, can be represented as the intersection of two wedges, and correspondingly, the set of local observables in a double cone is taken as the intersection of the operator sets associated with the wedges, on the level of von Neumann algebras. It is then a highly nontrivial task to show that this intersection contains more than just multiples of the identity; this problem was settled by Lechner \cite{Lechner:2008} using Tomita-Takesaki modular theory. From an abstract perspective, this constitutes a full construction of the model. However, while the existence of local observables is abstractly proven, very little is known about their concrete form, since the passage to von Neumann algebras hides an intricate limiting process that is difficult to trace.

Our aim here is to outline a programme that allows us to gain more insight into the structure of those local observables. The technical basis for this approach is a series expansion of the operators which is somewhat similar, but not identical, to the form factor expansion. Locality of operators can then be characterized via properties of the individual expansion terms.

In order to explain the idea, let us consider free field theory for a moment. In the theory of a real scalar free field, Araki has shown \cite{Ara:lattice} that any bounded operator $A$ on Fock space can be expanded into a series of normal-ordered annihilation and creation operators. For the two-dimensional massive free field and in a somewhat different notation from Araki, this expansion reads
\begin{equation}\label{eq:expansionfree}
  A = \sum_{m,n=0}^\infty \int \frac{d\thetav \, d\etav}{m!n!} f_{m,n}(\thetav,\etav) a^{\dagger}(\theta_1)\cdots a^\dagger(\theta_m) a(\eta_1)\cdots a(\eta_n),
\end{equation}
where $\ad,a$ are the usual annihilation and creation operators depending on rapidities $\theta_j$, $\eta_j$, and where the (generalized) functions $f_{m,n}$ are given by
\begin{equation}\label{eq:coefffree}
 f_{m,n}(\thetav,\etav) = \bighscalar{\Omega}{ [a(\theta_m), \ldots [a(\theta_1),[\ldots[A,\ad(\eta_n)]\ldots ,\ad(\eta_1)] \ldots ] \;\Omega  }.
\end{equation}
We will clarify technical details of the series in Sec.~\ref{sec:araki}. The expansion \eqref{eq:expansionfree} can be seen as a significant generalization of the well-known fact that every second-quantized operator can be written in the form $A=\int d\theta d\eta f(\theta,\eta) \ad(\theta)a(\eta)$. 

While the general expansion is valid for any $A$, whether localized or not, it is interesting to note that localization properties of $A$ are reflected in analyticity properties of $f_{m,n}$ and in bounds for this analytic continuation. In particular, if $A$ is localized in a bounded spacetime region, then the $f_{m,n}$ are entire analytic; this follows by writing $a,\ad$ in \eqref{eq:coefffree} in terms of Fourier transforms of time-zero fields. This and similar techniques have successfully been applied for establishing various phase space properties in the free field \cite{BucPor:phase_space,Bos:short_distance_structure,BostelmannDAntoniMorsella:2010}. 

Following a suggestion in \cite{Lashkevich:1994,SchroerWiesbrock:2000-1}, we propose to use a similar expansion for operators that is adapted to integrable models. In all what follows, we will restrict to theories that have the same particle spectrum as the real scalar free field. The proposed generalization of \eqref{eq:expansionfree} is then to replace the annihilators and creators $\ad,a$ with Zamolodchikov operators $\zd,z$ which depend on the given interaction:
\begin{equation}\label{eq:expansionz}
  A = \sum_{m,n=0}^\infty \int \frac{d\thetav \, d\etav}{m!n!} f_{m,n}(\thetav,\etav) z^{\dagger}(\theta_1)\cdots z^\dagger(\theta_m) z(\eta_1)\cdots z(\eta_n).
\end{equation}
Again, one would look for analyticity properties of the coefficients $f_{m,n}$ that reflect the locality of $A$.

In order to make these thoughts precise, our programme comprises the following steps. First, we need to prove that the ``deformed'' expansion \eqref{eq:expansionz} is valid for any $A$, and clarify the general properties of the expansion coefficients $f_{m,n}$, independent of a possible localization of $A$. Second, we want to characterize locality of an observable $A$ in terms of its expansion coefficients, giving necessary and sufficient conditions on $f_{m,n}$ that make $A$ local in a bounded region. A third step will be to  construct local observables $A$ explicitly by giving a sequence of expansion coefficients $f_{m,n}$, and showing that these fulfill the conditions mentioned previously. We note that we aim at examples of closable, possibly unbounded local operators $A$ here; but unlike the form factor programme, our goal are not Wightman $n$-point functions.

In the present paper, we deal only with the first mentioned problem. Leaving all aspects of locality aside, we will establish existence, uniqueness and general properties of the expansion \eqref{eq:expansionz}. We will deal with the characterization of locality \cite{BostelmannCadamuro:characterization-wip} and with concrete examples \cite{BostelmannCadamuro:examples-wip} elsewhere; see also \cite{Cadamuro:2012} for results in this direction.

Our present task is, in particular, to clarify the topological properties of the expansion: We need to show that every observable $A$ of a certain class can be expanded in a series as in \eqref{eq:expansionz}, and vice versa, that every set of (generalized) functions $f_{m,n}$ fulfilling certain regularity conditions defines via \eqref{eq:expansionz} an observable $A$ of the same class. Since the expansion itself involves the unbounded objects $z,\zd$, it is evident that the natural class of observables will not be bounded operators as in \cite{Ara:lattice}. Rather, we will establish the expansion for quadratic forms $A$ of a specific regularity class. These can be unbounded both at high particle numbers and at high energies, in a controlled way. We are thinking here in particular of the generalized $H$ bounds proposed by Jaffe \cite{Jaffe:1967}, although the present analysis is not restricted to these. In view of applications, one would also like to give conditions on the $f_{m,n}$ that guarantee an extension of the quadratic form $A$ to a closed, possibly unbounded operator. Sufficient conditions can in fact be found \cite[Ch.~4]{Cadamuro:2012}, and will be discussed elsewhere \cite{BostelmannCadamuro:examples-wip}.

Another important point is the behavior of the expansion coefficients under symmetry transformations of $A$. Specifically, we will see that the action of space-time reflections on the $f_{m,n}$ encodes the interaction of the model. While this aspect is interesting in its own right, it will also turn out to be crucial in the analysis of local observables in bounded regions \cite{BostelmannCadamuro:characterization-wip}.

We also want to investigate how the explicit formula \eqref{eq:coefffree} for the expansion coefficients generalizes to the interacting situation. A priori, this does not seem clear; while an explicit expression for $f_{m,n}$ can be given, see Eq.~\eqref{eq:fmndef} below, it bears little resemblance with the nested commutator formula that is valid in the free field situation. However, it turns out that in a certain class of models, Eq.~\eqref{eq:coefffree} can be generalized using a ``deformed commutator'' that depends on the interaction. To this end, we make use of the \emph{warped convolution} integral introduced by Buchholz, Summers and Lechner \cite{BuchholzSummers:2008,BuchholzSummersLechner:2011} which yields an alternative construction of certain 1+1 dimensional integrable models, using a deformation of the wedge-local observables. This also sheds light on the relation of our operator expansion to the deformation methods in quantum field theory recently used by several authors \cite{GrosseLechner:2007,Lechner:2011}.

The remainder of this article is organized as follows. We first introduce our mathematical setting, which is largely similar to \cite{Lechner:2008}, in Sec.~\ref{sec:preliminaries}. Then, in Sec.~\ref{sec:araki}, we establish the series expansion \eqref{eq:expansionz} and investigate its properties, in particular its behavior under spacetime symmetries. In Sec.~\ref{sec:warped}, we look at generalizations of the nested commutator formula \eqref{eq:coefffree}, and clarify relations with the warped convolution integral. We end with a brief outlook in Sec.~\ref{sec:conclusion}.

This article is based on the Ph.D.~thesis of one of the authors \cite{Cadamuro:2012}.

\section{Preliminaries}\label{sec:preliminaries}

We clarify some technical preliminaries and fix our notation, mostly following the setting of \cite{Lechner:2008}.
Throughout the paper, our spacetime is 1+1 dimensional Minkowski space $\rbb^2$ with bilinear form $x \cdot y = x_0 y_0-x_1 y_1$. The quantum field theoretic model that we will consider depends on a scattering function $S$ as a parameter, details of which will be discussed in a moment. We remark that for the special case $S=1$, all structures in the following reduce to the well-known situation on the Fock space of a free real scalar Bose field.
 
\subsection{Scattering functions}\label{sec:scatter}

We will first explain the properties of the two-particle scattering matrix of our integrable models. As in \cite{Lechner:2008}, we will restrict to theories with only one species of scalar particle. The two-particle scattering matrix is then just a complex-valued function, $S$, depending on a rapidity difference. Mathematically, we take it to be a smooth function $S: \rbb \to \cbb$ with the property
\begin{equation}
  \forall \theta\in \rbb: \quad S(\theta)^{-1}=S(-\theta)=\overline{S(\theta)}.
\end{equation}
It is clear that $|S(\theta)|=1$. Since we disregard all aspects of locality in this paper, we do not require an analytic or meromorphic continuation of $S$, and the properties usually demanded of this continuation, such as crossing symmetry, are not relevant for the present analysis.

For our purposes, a representation of the permutation group $\perms{n}$ of $n$ elements, related to the scattering function $S$, plays a crucial role. For any $\sigma \in \perms{n}$, we consider the following function $S^\sigma$ on $\rbb^n$,
\begin{equation}\label{eq:Sperm}
S^\sigma (\thetav) := \prod_{\substack{i<j \\ \sigma(i)>\sigma(j)}} S(\theta_{\sigma(i)}-\theta_{\sigma(j)}).
\end{equation}
These functions fulfill the composition law (cf.~\cite[p.~54]{Lechner:2006})
\begin{equation}\label{eq:Scompose}
 S^{\sigma\circ\rho} (\thetav) = S^{\sigma}(\thetav) S^\rho(\thetav^\sigma),
\end{equation}
where $\thetav^{\sigma}=(\theta_{\sigma(1)},\ldots,\theta_{\sigma(n)})$. We can then introduce an action $D_{n}$ of $\perms{n}$ on $L^{2}(\mathbb{R}^{n})$ by
\begin{equation}
(D_{n}(\sigma)f)(\thetav)= S^{\sigma}(\thetav)f(\thetav^{\sigma}),\quad \sigma \in \perms{n}.
\end{equation}
We use the same symbol for the corresponding actions on the space of compactly supported test functions, $\dcal(\rbb^n):=\ccal_0^\infty(\rbb^n)$, and on its dual, the space of distributions $\dcal(\rbb^n)'$. Using the composition law \eqref{eq:Scompose}, it follows that $D_{n}$ defines a group representation of $\perms{n}$ on those spaces, unitary in the case of $L^{2}(\mathbb{R}^{n})$.  
Then $P^{S}_{n}:=\frac{1}{n!}\sum_{\sigma \in \perms{n}}D_{n}(\sigma)$ is a projection onto the space of $D_n$-invariant, or \emph{$S$-symmetric}, functions. (In the $L^2$ case, $P^{S}_{n}$ is the unique orthogonal projection.) A function $f$ is $S$-symmetric if and only if
\begin{equation}
\forall \sigma \in \perms{n}: \quad f(\thetav)=S^{\sigma}(\thetav)f(\thetav^{\sigma});
\end{equation}
due to the representation property, it suffices to check this condition on transpositions $\sigma$. We alternatively write $P_n^S f (\thetav) = \operatorname{Sym}_S f (\thetav)$, and call $\operatorname{Sym}_S f$ the \emph{$S$-symmetric part} of $f$. If the function depends on several variables and we want to take the $S$-symmetric part only with respect to some of them, we will denote this as in $\operatorname{Sym}_{S,\thetav} f(\thetav,\thetav')$. The choice of variables for $S$-symmetrization can be of importance, as the formula
\begin{equation}\label{symstheta}
\operatorname{Sym}_{S,\thetav}\delta^{n}(\thetav-\thetav')=
\operatorname{Sym}_{S^{-1},\thetav'}\delta^{n}(\thetav-\thetav')
\end{equation}
shows.

\subsection{Hilbert space}\label{sec:hilbertspace}

As mentioned, we will focus our attention on models with only one species of scalar particle with mass $\mu >0$. As in the free real scalar field, our single particle space is then $\Hil_{1}=L^{2}(\mathbb{R},d\theta)$, where $\theta$ (``rapidity'') is related to the particle momentum by
\begin{equation}
p(\theta):=\mu \begin{pmatrix} \cosh\theta \\ \sinh\theta \end{pmatrix}, \quad \theta\in \mathbb{R}.
\end{equation}
Using the subspace of $S$-symmetric wave functions as introduced in Sec.~\ref{sec:scatter}, we define the $n$-particle space as $\Hil_{n}:= \operatorname{Sym}_S \Hil_{1}^{\otimes n}$, with $\Hil_0=\cbb\Omega$. The Hilbert space $\Hil$ of the theory is then the ``$S$-symmetrized Fock space'' over $\Hil_{1}$:
\begin{equation}
\Hil:=\bigoplus_{n=0}^{\infty}\Hil_{n}.
\end{equation}
We denote the orthogonal projection onto $\Hil_n\subset \Hil$ with $P_n$, and define $\fpnp_n:=\sum_{j=0}^{n}P_j$. Further, we denote the space of finite particle number states with $\fpn := \bigcup_{n}\fpnp_{n}\Hil$; it is dense in $\Hil$.

The space-time symmetry group of our 1+1-dimensional system is the proper Poincar\'e group, generated by space-time translations, a one-parameter family of boosts, and the space-time reflection. These symmetries act on $\Hil$ via a strongly continuous, (anti)unitary representation $U$ as follows. 
Space-time translations and boosts act on $\psi = \oplus_{n=0}^\infty \psi_n\in\Hil$ as
\begin{equation}\label{eq:translboostdef}
(U(x,\lambda)\psi)_{n}(\thetav):= e^{ i p(\thetav) \cdot x } \psi_{n}(\thetav-\lambdav),
\quad \text{where }\;p(\thetav) = \sum_{k=1}^{n}p(\theta_{k}), \;\; \lambdav = (\lambda,\ldots,\lambda),
\end{equation}
while the space-time reflection acts by an antiunitary operator $U(j)=:J$ as
\begin{equation}
(U(j)\psi)_{n}(\thetav):= \overline{\psi_{n}(\theta_{n},\ldots,\theta_{1})}.
\end{equation}
As usual, we denote the positive generator of time translations as $H$.

Apart from bounded operators on $\Hil$, we will also need to deal with unbounded quadratic forms. Their unboundedness can relate to their behavior at high particle numbers, but also at high energies, in a controlled way. Roughly speaking, we will allow high energy behavior like $e^{\omega(H/\mu)}$, where $\omega$ is a positive function, called the \emph{indicatrix}. In particular, we are thinking here of $\omega(p)$ growing slightly less than linearly in $p$, as proposed by Jaffe \cite{Jaffe:1967} for a high-energy behavior that is compatible with locality of operators. However, since for the present discussion we ignore all aspects of locality, we can allow a very generic indicatrix: In the following, we will only assume that $\omega: [0,\infty)\to [0,\infty)$ is smooth, monotonously increasing, and sublinear, i.e.,
\begin{equation}
 \forall p,q\geq 0:\quad \omega(p+q)\leq \omega(p)+\omega(q).
\end{equation}
Related to a fixed indicatrix $\omega$, we introduce the following subspaces of our Hilbert space $\Hil$. We denote $\Hil^\omega := \{\psi \in \Hil : \|e^{\omega(H/\mu)} \psi\| < \infty\}$. For fixed $n$, we write $\Hil^\omega_n := \Hil^\omega \cap \Hil_n$, and $\fpno := \Hil^{\omega} \cap \fpn$.  We note that $\fpno \subset \Hil^\omega$ is dense in $\Hil$. Also, let us set for test functions $g \in \dcal(\rbb^n)$,
\begin{equation}
  \onorm{g}{2} := \gnorm{\thetav \mapsto e^{\omega(E(\thetav))}g(\thetav)}{2},
\end{equation}
where $E$ is the dimensionless energy function,
\begin{equation}
   E(\thetav) := p_0(\thetav)/\mu = \sum_{j=1}^n \cosh \theta_j .
\end{equation}

We can now formalize the set of quadratic forms of interest. By $\qf^\omega$, we denote the space of quadratic (more precisely, sesquilinear) forms $A$ on $\fpno\times\fpno$, that is,
\begin{equation}
   A: \fpno \times \fpno \to \cbb, \quad  (\psi,\chi) \mapsto \hscalar{\psi}{ A \chi},
\end{equation}
such that the following norm is finite for any $n \in \nbb_0$:
\begin{equation}\label{eq:aomeganorm}
 \gnorm{A}{n}^{\omega} := \frac{1}{2} \gnorm{\fpnp_n A e^{-\omega(H/\mu)}\fpnp_n}{} + \frac{1}{2} \gnorm{\fpnp_n e^{-\omega(H/\mu)} A \fpnp_n}{}.
\end{equation}
We note that space-time translations and reflections act on $\qf^\omega$ by adjoint action of $U(x,0)$ and $J$, respectively, since these operators commute with $H$. The adjoint action of Lorentz boosts $U(0,\lambda)$ does not necessarily leave $\qf^\omega$ invariant, but maps $\qf^\omega$ into $\qf^{\omega'}$ with $\omega'(p) = \omega(cp)$, and where $c>0$ is chosen for fixed $\lambda$ such that $U(0,\lambda) H U(0,\lambda)^\ast \leq c H$. (It would be possible to modify the definition of $\qf^\omega$ so that it becomes fully Poincar\'e invariant, but we stick to the present definition for reasons of simplicity.)

\subsection{Generalized annihilation and creation operators} \label{sec:zgen}

Another crucial ingredient to our analysis is the representation of the Zamolodchikov-Faddeev algebra given by $z,\zd$, which is constructed as a deformed version of the usual CCR algebra as follows \cite{Lechner:2003}. For $f\in\Hil_1$, the operators $z(f)$, $\zd(f)$ are defined on $\fpn$ by 
\begin{align}
(\zd(f)\psi)_{n} &:= \sqrt{n}\,P^{S}_{n}(f\otimes \psi_{n-1}),\\
(z(f) \psi)_n &:= \sqrt{n+1} \int d\theta \, f(\theta)\, \psi_{n+1}(\theta,\cdotarg),
\end{align}
where $\psi=\oplus_{n} \psi_n \in \fpn$.  
These ``smeared'' annihilators and creators $\zd(f)$, $z(f)$ are adjoints of each other; more precisely, $z(f) = \zd(\overline{f})^{*}\restrict \fpn$. They are unbounded operators on $\fpn$, but their bounds on fixed particle number vectors can be controlled: In generalization of \cite[Eq.~(3.14)]{Lechner:2008}, we have for $n \in \nbb_0$ and $f \in \Hil_1$,
\begin{equation}\label{omegaz}
 \gnorm{ e^{\omega(H/\mu)} \zd(f) e^{-\omega(H/\mu)} \fpnp_n }{}
  \leq \sqrt{n+1} \onorm{f}{2},
\quad
 \gnorm{ e^{\omega(H/\mu)} z(f) e^{-\omega(H/\mu)} \fpnp_n }{}
  \leq \sqrt{n} \onorm{f}{2} 
\end{equation}
if the right-hand side is finite. (Monotonicity and sublinearity of $\omega$ enter here.)
The $z(f)$, $\zd(f)$ are linear in $f$, and we will often write them as integrals of formal kernels\footnote{%
Throughout this paper, we will usually denote distributions as integrals of formal kernels. This should be understood merely as a notational convention. It is convenient for us since in applications \cite{BostelmannCadamuro:characterization-wip}, many distributions will arise as boundary values of analytic functions.
}, $z^\#(f) = \int f(\theta) z^\#(\theta) d\theta$. In the sense of distributions, these kernels act on a wave function $\psi \in \hcal_n$ as
\begin{align}
(\zd(\theta)\psi)(\lambdav) &= \sqrt{n+1} \,\operatorname{Sym}_{S,\lambdav} \delta(\theta-\lambda_1)\psi(\lambda_2,\ldots,\lambda_{n+1}),\\
(z(\eta) \psi)(\lambdav) &= \sqrt{n} \,\psi(\eta,\lambdav).
\end{align}
They fulfill the relations of the Zamolodchikov-Faddeev algebra:
\begin{align}
\zd(\theta)\zd(\theta') &= S(\theta-\theta')\zd(\theta')\zd(\theta),\\
z(\eta)z(\eta') &= S(\eta-\eta')z(\eta')z(\eta),\\
z(\eta)\zd(\theta) &= S(\theta-\eta)\zd(\theta)z(\eta) + \delta(\theta-\eta) \idop.
\end{align}

In order to establish the expansion \eqref{eq:expansionz}, we require a multilinear extension of normal ordered products of the $z,\zd$, formally given by 
\begin{equation}
  z^{\dagger m} z^n(f) = \int d \thetav\, d \etav f(\thetav,\etav) \underbrace{\zd(\theta_1)\ldots \zd(\theta_m)z(\eta_1)\ldots z(\eta_n)}_{ =: z^{\dagger m} (\thetav) z^n (\etav)}.
\end{equation}
This is covered by our definitions so far if $f$ is of the form $f(\thetav,\etav)=f_1(\theta_1)\cdots f_{m+n}(\eta_n)$, or is a linear combination of such functions. Lechner \cite[Lemma 4.1.2]{Lechner:2006} extended the definition to arbitrary $f \in L^2(\rbb^{m+n})$.
We will need an even more general class of ``smearing functions'' $f$. To that end, 
for a distribution $f\in \dcal(\rbb^{m+n})'$, we introduce the (possibly infinite) norms\footnote{%
In a slight abuse of notation, here and in the following we will understand arguments like $e^{-\omega(E(\thetav))}f(\thetav,\etav)$ of norms to be distribution kernels
in the variables $\thetav$ and $\etav$, rather than function values at fixed $\thetav,\etav$.}
\begin{align}\label{eq:crossnorm}
\gnorm{f}{m \times n} &:= \sup \Big\{ \big\lvert \! \int f(\thetav,\etav) g(\thetav) h(\etav)  d\thetav d\etav \,\big\rvert : g \in \dcal(\rbb^m), \, h \in \dcal(\rbb^n), \,  \gnorm{g}{2} \leq 1, \, \gnorm{h}{2} \leq 1\Big\},
\\ \label{eq:omegacrossnorm}
\onorm{f}{m \times n} &:= \frac{1}{2} \gnorm{e^{-\omega(E(\thetav))}f(\thetav,\etav)}{m \times n} 
+ \frac{1}{2} \gnorm{f(\thetav,\etav)e^{-\omega(E(\etav))}}{m \times n}.
\end{align}
Here $\gnorm{f}{m \times n}$ can alternatively be understood as the operator norm of $f$ as the kernel of an integral operator from $L^2(\rbb^m)$ to $L^2(\rbb^n)$.
Let us note a few computation rules for the norms above. First, if $f_L\in \ccal^\infty(\rbb^m)$, $f_R\in \ccal^\infty(\rbb^n)$ are bounded, then
\begin{equation} \label{eq:mnnormfactor}
  \onorm{f_L(\thetav)f(\thetav,\etav)f_R(\etav)}{m \times n} \leq \gnorm{f_L}{\infty}  \onorm{f}{m \times n} \gnorm{f_R}{\infty}.
\end{equation}
(Namely, we can absorb $f_L, f_R$ into the test functions $g,h$, respectively.) Second, if $f \in \dcal(\rbb^{m+n})'$, $f' \in \dcal(\rbb^{m'+n'})'$, and if $f \cdot f' \in \dcal(\rbb^{m+m'+n+n'})'$ denotes the product in \emph{independent} variables, then 
\begin{equation} \label{eq:mnnormproduct}
   \onorm{f \cdot f'}{(m+m')\times (n+n')} \leq \onorm{f}{m \times n} \gnorm{f'}{m' \times n'}.
\end{equation}
(For $\omega=0$, this follows from the corresponding estimate for the norm of the tensor product of the associated integral operators \cite[Prop.~2.6.12]{KadRin:algebras1}. The general case can then be deduced from \eqref{eq:omegacrossnorm}, \eqref{eq:mnnormfactor} and monotonicity of $\omega$.)
Finally, it is clear that $\gnorm{f}{m \times n} \leq \gnorm{f}{2}$ and hence
\begin{equation}\label{eq:crossnormcomparison}
    \gnorm{f}{m \times n}^{\omega} \leq \frac{1}{2}\Big(\gnorm{e^{-\omega(E(\thetav))}f(\thetav,\etav)}{2}+\gnorm{e^{-\omega(E(\etav))}f(\thetav,\etav)}{2}\Big)
\end{equation}
if the right hand side is finite. However, equality does in general not hold: For $\omega=0$ and  $f(\theta_1,\theta_2) = \delta(\theta_1-\theta_2)$, we have $\onorm{f}{1 \times 1}=1$ but $f$ does not have finite $L^2$ norm.

We now define our multilinear annihilation and creation operators as follows.
For an arbitrary distribution $f \in \dcal(\rbb^{m+n})'$ and with vectors $\psi \in \Hil_k \cap \dcal(\rbb^k)$, $\chi \in \Hil_\ell \cap \dcal(\rbb^\ell)$, we set:
\begin{equation}\label{eq:zzqf}
   \hscalar{\chi}{ z^{\dagger m} z^n(f) \psi} :=
    \frac{\sqrt{k!(k-n+m)!}}{(k-n)!}
   \int d\lambdav\, d \thetav\, d \etav \, \overline{ \chi(\thetav,\lambdav)} f(\thetav,\etav) \psi(\eta_n\ldots\eta_1,\lambdav)
\end{equation}
if $\ell=k-n+m$ and $k\geq n$, and $=0$ otherwise. Because of the relation $(z(\eta)\psi)(\thetav)=\sqrt{k} \psi(\eta,\thetav)$, this extends the previous definition of the annihilators and creators. The question is now whether the quadratic form \eqref{eq:zzqf} can be extended to $\fpno \times \fpno$, or even to an (unbounded) operator on $\fpno$. A sufficient condition for that is $\gnorm{f}{m \times n}^\omega < \infty$, as the following proposition shows.

\begin{proposition} \label{proposition:zzdcrossnorm}

If $f \in \dcal(\rbb^{m+n})'$ with $\|f\|^{\omega}_{m \times n} < \infty$, then $z^{\dagger m} z^n(f)$ extends to an operator on $\fpno$, and for any $k \geq n$,
\begin{equation}
\big\|  z^{\dagger m} z^n(f)e^{- \omega(H/\mu)} \fpnp_k \big\| \leq 2 \frac{\sqrt{k!(k-n+m)!}}{(k-n)!}\|f\|^{\omega}_{m\times n}.\label{eq:cross}
\end{equation}
Moreover, for any $k \geq m,n$,
\begin{equation}
\gnorm{ z^{\dagger m} z^n(f) }{k}^\omega \leq 2 \frac{k!}{(k-\max(m,n))!}\|f\|^{\omega}_{m\times n}.\label{eq:zzdanorm}
\end{equation}
\end{proposition}

\begin{proof}

For $\psi \in\Hil_{k} \cap \dcal(\rbb^k)$ and $\chi\in \Hil_{\ell}\cap \dcal(\rbb^\ell)$, with $\ell=k-m+n$, one has from \eqref{eq:zzqf},
\begin{equation}\label{eq:zzmatrixelm}
\begin{aligned}
\left| \hscalar{ \chi }{  z^{\dagger m} z^n(f) \psi } \right|
& \leq \frac{\sqrt{k!(k-n+m)!}}{(k-n)!}\int d\lambdav \left| \int d\thetav d\etav\; \overline{\chi(\thetav,\lambdav)}\psi(\eta_n\ldots\eta_1,\lambdav)f(\thetav,\etav)\right|\\
& \leq  2 \frac{\sqrt{k!(k-n+m)!}}{(k-n)!} \gnorm{f}{m \times n}^{\omega} \gnorm{\chi}{2} \Big(\int d\lambdav\, d\etav\; |\psi(\etav,\lambdav)|^{2}e^{2\omega(E(\etav))} \Big)^{1/2} ,
\end{aligned}
\end{equation}
where we used \eqref{eq:omegacrossnorm} and the Cauchy-Schwarz inequality. Employing monotonicity of $\omega$ in the remaining integrand, we arrive at
\begin{equation}
\left| \hscalar{ \chi }{  z^{\dagger m} z^n(f) \psi } \right|
 \leq 2 \frac{\sqrt{k!(k-n+m)!}}{(k-n)!}\onorm{f}{m \times n}  \gnorm{\chi}{2} \onorm{\psi}{2}.
\end{equation}
Since $\psi$ and $\chi$ were chosen from dense sets in the corresponding  spaces, and since the matrix elements \eqref{eq:zzmatrixelm} vanish if $\ell \neq k-n+m$, we can extend $z^{\dagger m} z^n(f)$ to a bounded operator on $\hcal_k^\omega$ with norm
\begin{equation}\label{eq:pkbound}
\left\|z^{\dagger m} z^n(f)P_{k}e^{-\omega(H/\mu)}  \right\|\leq 2\frac{\sqrt{k!(k-n+m)!}}{(k-n)!}\gnorm{f}{m\times n}^{\omega}.
\end{equation}
This works for any $k$. For $k\neq k'$, the images of $z^{\dagger m} z^n(f)P_{k}$ and $z^{\dagger m} z^n(f)P_{k'}$ are orthogonal; thus \eqref{eq:cross} follows from \eqref{eq:pkbound} using Pythagoras' theorem.

For \eqref{eq:zzdanorm}, we deduce from \eqref{eq:cross} that
\begin{equation}\label{eq:zzdenergynorm}
 \gnorm{\fpnp_k z^{\dagger m} z^n(f) e^{-\omega(H/\mu)} \fpnp_k}{} \leq 2 \frac{k!}{(k-\max(m,n))!}  \gnorm{f}{m\times n}^{\omega};
\end{equation}
this is best seen when considering the cases $m > n$ and $m \leq n$ separately.
Further, we note that, in the sense of quadratic forms, $(z^{\dagger m} z^n(f))^\ast = z^{\dagger n} z^m(f^\ast)$, where $f^\ast(\thetav,\etav) = \overline{f(\eta_m,\ldots,\eta_1,\theta_n,\ldots,\theta_1)}$, and where one finds $\onorm{f^\ast}{n \times m} = \onorm{f}{m \times n}$.
An application of \eqref{eq:zzdenergynorm} then yields
\begin{equation}
 \gnorm{\fpnp_k e^{-\omega(H/\mu)} z^{\dagger m} z^n(f)  \fpnp_k}{}
 = \gnorm{\fpnp_k z^{\dagger n} z^m(f^\ast) e^{-\omega(H/\mu)} \fpnp_k}{}
\leq 2 \frac{k!}{(k-\max(m,n))!}  \gnorm{f}{m\times n}^{\omega},
\end{equation}
and thus \eqref{eq:zzdanorm} is proven.
\end{proof}

\section{The operator expansion} \label{sec:araki}

In this section, we establish the proposed series expansion \eqref{eq:expansionz} in a rigorous fashion; the main result is Theorem~\ref{theorem:arakiexpansion}. As announced, we also deal with the action of symmetry transformations on the expansion coefficients.

\subsection{Contracted matrix elements}

As a first step, we explain our notation for \emph{contractions}, similar to \cite{Lechner:2008} but with slightly different conventions. Contractions are devices for denoting variants of the formal matrix element
\begin{equation}
 \hscalar{\zd(\theta_1) \cdots \zd(\theta_m) \Omega}{A \zd(\eta_n) \cdots \zd(\eta_1) \Omega}
  =: \hscalar{\lvector{}{\thetav} }{ A \rvector{}{\etav}}.
\end{equation}
A contraction $C$ is a triple, $C=(m,n,\{(l_1,r_1),\ldots,(l_{k},r_k)\})$, where $m,n\in\nbb_0$, $1 \leq l_j \leq m$ and $m+1 \leq r_j \leq m+n$, and both the $l_j$ and the $r_j$ are pairwise different among each other. We denote $\ccal_{m,n}$ the set of all contractions for fixed $m$ and $n$, and write $|C|:=k$ for the \emph{length} of the contraction; here we allow the case $|C|=0$ (the \emph{empty contraction}). The $(l_j,r_j)$ will be called pairs of contracted indices. Contractions are used to label modified (``contracted'') matrix elements $\hscalar{\lvector{C}{\thetav}}{ A \rvector{C}{\etav}}$,
where
\begin{align}
   \lvector{C}{\thetav} &:= \zd(\theta_1) \cdots \widehat{\zd(\theta_{l_1})} \cdots \widehat{\zd(\theta_{l_{|C|}})} \cdots \zd(\theta_m) \Omega,
\\
   \rvector{C}{\etav} &:= \zd(\eta_n) \cdots \widehat{\zd(\eta_{r_1-m})} \cdots \widehat{\zd(\eta_{r_{|C|}-m})} \cdots \zd(\eta_1) \Omega,
\end{align}
and where the hats indicate that the marked elements have been left out of the sequence. 

At this point, a remark about the well-definedness of $\hscalar{\lvector{C}{\thetav}}{ A \rvector{C}{\etav}}$ as a distribution is in order. First, one can understand
 $\lvector{C}{\cdotarg}$ as an $\Hil$-valued distribution with values in $\fpno$. More precisely, for $f\in\dcal(\rbb^{m-|C|})$, and with $\hat\thetav\in\rbb^{m-|C|}$ denoting $\thetav$ with the components $\theta_{\ell_j}$ left out, the expression $\lvector{C}{f}=\int d\hat\thetav \, f(\hat\thetav) \lvector{C}{\thetav}$ defines a vector in $\hcal_{m-|C|}$, and
\begin{equation}\label{eq:lcnorm}
   \gnorm{e^{\omega(H/\mu)} \lvector{C}{f}}{} \leq \sqrt{(m-|C|)!}\onorm{f}{2},
\end{equation}
cf.~\eqref{omegaz}. This holds for $\rvector{C}{\cdotarg}$ analogously. 
Now let $A \in \qf^\omega$. In view of \eqref{eq:aomeganorm} and \eqref{eq:lcnorm}, the map $(f,g) \mapsto \hscalar{\lvector{C}{\bar f}}{ A \rvector{C}{g}}$ is well-defined and continuous in the topology of $\dcal(\rbb^{m-|C|}) \times \dcal(\rbb^{n-|C|})$. Hence, we can understand 
$\hscalar{\lvector{C}{\thetav}}{ A \rvector{C}{\etav}}$ as the kernel of this distribution.

With a contraction $C\in \ccal_{m,n}$, we associate the quantities
\begin{align}
\delta_{C}(\thetav,\etav) &:= \prod_{j=1}^{|C|}\delta(\theta_{l_j}-\eta_{r_j-m}), \label{eq:deltac}\\
 S_C(\thetav,\etav) &:= \Big(\prod_{j=1}^{|C|}\prod_{p_{j}=l_{j}+1}^{r_{j}-1}S^{(m)}_{p_{j},l_{j}}\Big) \prod_{\substack{r_{i}<r_{j} \\ l_{i}<l_{j}}} S^{(m)}_{l_{j},r_{i}},\label{eq:sc}
\end{align}
where we used the following shorthand notation for $\xiv\in\rbb^{m+n}$,
\begin{equation}
S_{a,b}(\xiv):=S(\xi_{a}-\xi_{b}),\quad 
 S_{a,b}^{(m)} := \begin{cases} S_{b,a}(\xiv) \quad \text{if }a \leq m <b \text{ or } b\leq m <a, \\S_{a,b} (\xiv)\quad \text{otherwise.}\end{cases} 
\end{equation}
We will often leave away the arguments $\thetav,\etav$ where they are clear from the context. We will see the use of the above expressions later.

It is a useful fact that we can express the factors $S_C$ in terms of the expressions $S^\sigma$ associated with permutations $\sigma$.

\begin{lemma}\label{lemma:scpermute}
 There holds
\begin{equation}\label{Scontrperm}
 \delta_{C} S_C(\thetav,\etav) = \delta_{C}S^{\sigma}(\thetav)S^{\rho}(\etav),
\end{equation}
where
\begin{equation}\label{eq:scpermutations}
 \begin{aligned}
\sigma=\begin{pmatrix} 1 & & \ldots & &  & m \\ 1 & \ldots \;\hat{l}\; \ldots & m &  l_{1} & \ldots & l_{|C|} \end{pmatrix}, \quad
\rho=\begin{pmatrix} m+1 & & & \ldots & & m+n \\ r_{|C|} & \ldots & r_{1} & m+1 & \ldots \; \hat{r}\; \ldots & m+n \end{pmatrix}.
 \end{aligned}
\end{equation}
\end{lemma}
\noindent
\emph{Remarks}: $\hat l$ stands for leaving out the $l_j$ from the sequence; $\hat r$ analogously. The permutations $\sigma,\rho$ are not unique since one can permute the pairs of contracted indices. However, due to the factor $\delta_C$, the right hand side of \eqref{Scontrperm} is independent of this choice.

\begin{proof}

By the above remark, we can assume that $r_{1} < \ldots < r_{|C|} $. From the definition of $S^\sigma$, $S^\rho$ in Eq.~\eqref{eq:Sperm}, we read off that
\begin{equation}
 S^{\sigma} = \prod_{j=1}^{|C|}\prod_{p_{j}=l_{j}+1}^{m}S_{p_{j},l_{j}} \cdot \prod_{\substack{i<j \\ l_{i}<l_{j}}} S_{l_{i},l_{j}}, \quad
 S^{\rho}   = \prod_{j=1}^{|C|}\prod_{q_{j}=m+1}^{r_{j}-1}S_{r_{j},q_{j}}.
\end{equation}
Taking the factor $\delta_C$ into account, a short computation shows that
$\delta_{C} S^{\sigma} S^{\rho} = \delta_{C} S_C$ with $S_C$ defined as in \eqref{eq:sc}.
\end{proof}

We will often need to consider \emph{compositions} of contractions. For $C \in \ccal_{m,n}$ and $C' \in \ccal_{m-|C|,n-|C|}$, we denote the composed contraction -- with indices contracted first with $C$, then with $C'$ -- as $C \dot \cup C' \in \ccal_{m,n}$. (The definition should intuitively be clear, but note that it involves a \emph{renumbering} of the indices in $C'$ before taking the set union; we will often avoid to indicate the renumbering explicitly.) The factors $\delta_C$ and $S_C$ compose as follows in this situation. Here $\hat \thetav \in \rbb^{m-|C|}$ is $\thetav$ with the components $\theta_{l_1},\ldots,\theta_{l_{|C|}}$ left out; analogously for $\hat \etav$.

\begin{lemma}\label{lemma:contractcompose}

Let $C \in \ccal_{m,n}$ and $C'\in \ccal_{m-|C|,n-|C|}$. There holds
\begin{equation}
\delta_{C}(\thetav,\etav) \delta_{C'} (\hat\thetav,\hat\etav) S_C(\thetav,\etav) S_{C'}(\hat\thetav,\hat\etav) = \delta_{C\dot{\cup}C'}(\thetav,\etav) S_{C\dot{\cup}C'}(\thetav,\etav) \label{scc}.
\end{equation}
\end{lemma}

\begin{proof}
It is clear from the definition \eqref{eq:deltac} that $\delta_{C}\delta_{C'}=\delta_{C \dot\cup C'}$.
Using this fact and Lemma~\ref{lemma:scpermute}, it remains to show that
\begin{equation}\label{Scomposperm}
  S^{\sigma}(\thetav)S^{\sigma'}(\hat \thetav ) = S^{\sigma''}(\thetav),
\quad
  S^{\rho}(\etav) S^{\rho'}(\hat \etav ) = S^{\rho''}(\etav),
\end{equation}
where $\sigma,\rho$, $\sigma',\rho'$, $\sigma'',\rho''$ are permutations associated by Eq.~\eqref{Scontrperm} with $C$, $C'$, and $C\dot\cup C'$, respectively.
We note that $\sigma'$ is given explicitly by
\begin{equation}
\sigma'=\begin{pmatrix} 1 & & \ldots \hat{l} \ldots &  & & m \\ 1 & \ldots \hat{l}\; \hat{l'} \ldots & m &  l'_{1} & \ldots & l'_{|C'|} \end{pmatrix} \in \perms{m-|C|}.
\end{equation}
We can consider $\sigma'$ as an element of $\perms{m}$ by extending the permutation matrix as follows:
\begin{equation}
\sigma'=\begin{pmatrix} 1 & & \ldots \hat{l} \ldots & &  & m  & l_{1} & \ldots & l_{|C|} \\ 1 & \ldots \hat{l'}\; \hat{l} \ldots & m  &  l'_{1} & \ldots & l'_{|C'|} & l_{1} & \ldots & l_{|C|}  \end{pmatrix}.
\end{equation}
With this identification, $S^{\sigma'}(\hat\thetav)=S^{\sigma'}(\thetav^{\sigma})$, and one notes that $\sigma'' := \sigma \circ \sigma'$ is indeed associated with $C \dot\cup C'$ by Eq.~\eqref{Scontrperm}. By the composition law in Eq.~\eqref{eq:Scompose}, one has
\begin{equation}
 S^{\sigma''}(\thetav) = S^{\sigma}(\thetav) S^{\sigma'}(\thetav^\sigma) = S^{\sigma}(\thetav) S^{\sigma'}(\hat\thetav).
\end{equation}
Analogously one obtains the second part of Eq.~\eqref{Scomposperm}, and hence the result.
\end{proof}

For any quadratic form $A \in \mathcal{Q}^{\omega}$, we now define its \emph{fully contracted matrix elements}
$\cme{m,n}{A}$ by
\begin{equation}\label{eq:fmndef}
\cme{m,n}{A}(  \thetav,  \etav) :=
\sum_{C\in \ccal_{m,n}} (-1)^{|C|} \delta_C \, S_C (\thetav,\etav) \,
\hscalar{ \lvector{C}{\thetav} }{ A \,\rvector{C}{\etav} }.
\end{equation}
These $\cme{m,n}{A}$ will turn out to be the expansion coefficents in the expansion \eqref{eq:expansionz}. They are very similar to the quantities $\langle \cdotarg \rangle_{m+n,m}^\text{con}$ introduced in \cite{Lechner:2008};
in notation used there, we have $\cme{m,n}{A} = \langle J A^\ast J \rangle_{m+n,m}^\text{con}$.

It is evident that $\cme{m,n}{A}$ are distributions in $\dcal(\rbb^{m+n})'$. However, we can show more: the norms $\gnorm{\cme{m,n}{A}}{m\times n}^{\omega}$ are finite.

\begin{proposition}\label{proposition:fmnbound}
 For $m,n \in \nbb_0$, there is a constant $c_{mn}$ such that for all $A \in \qf^{\omega}$,
\begin{equation}
\onorm{\cme{m,n}{A} }{m\times n} \leq c_{mn} \onorm{A}{m+n}.
\end{equation}
\end{proposition}

\begin{proof}
The triangle inequality yields
\begin{equation}
\onorm{\cme{m,n}{A} }{m\times n} \leq \sum_{C\in \ccal_{m,n}} \onorm{\delta_C  S_C (\thetav,\etav)\,
\hscalar{ \lvector{C}{\thetav} }{ A \,\rvector{C}{\etav} } }{m \times n} .
\end{equation}
The factor $S_C (\thetav,\etav)$ factorizes to $S^\sigma(\thetav)S^\rho(\etav)$, see Lemma~\ref{lemma:scpermute}, and therefore can be estimated by 1, cf.~Eq.~\eqref{eq:mnnormfactor}. The individual factors of $\delta_C(\thetav,\etav)$ and the matrix element $\hscalar{ \lvector{C}{\thetav} }{ A \,\rvector{C}{\etav} } $ depend on mutually different variables; repeated application of \eqref{eq:mnnormproduct} then gives
\begin{equation} \label{eq:mntriangle}
\onorm{\cme{m,n}{A} }{m\times n} \leq \sum_{C\in \ccal_{m,n}} \Big(\prod_{j=1}^{|C|} \gnorm{\delta(\theta_{l_j}-\eta_{r_j})}{1 \times 1} \Big) \onorm{
\hscalar{ \lvector{C}{\thetav} }{ A \,\rvector{C}{\etav} } }{(m-|C|) \times (n-|C|)} .
\end{equation}
One easily sees $\gnorm{\delta(\theta-\eta)}{1 \times 1}=1$, and
\begin{equation}
 \onorm{\hscalar{ \lvector{C}{\thetav} }{ A \,\rvector{C}{\etav} } }{(m-|C|) \times (n-|C|)}
 \leq \sqrt{(m-|C|)!}\,\sqrt{(n-|C|)!} \;\onorm{A}{m+n}
\end{equation}
by Cauchy-Schwarz. Inserted into \eqref{eq:mntriangle}, this gives the result.
\end{proof}

As an important input to the following construction, we now show that the $\cme{m,n}{A}$ are always $S$-symmetric. We note that this is an improvement over \cite{Lechner:2008}: There, the $S$-symmetry was known only for $n=0$ and for $m=0$, or as a consequence of an analytic continuation due to locality of $A$.

\begin{proposition}\label{proposition:fmnsymm}
 Let $A\in\qf^\omega$. The $\cme{m,n}{A}$ are $S$-symmetric in the first $m$ and last $n$ variables separately; that is, for permutations
  $\pi \in \perms{m}$ and $\tau \in \perms{n}$,
\begin{equation}
   \cme{m,n}{A}(\thetav,\etav) = S^\pi(\thetav) S^\tau(\etav)\cme{m,n}{A}(\thetav^\pi,\etav^\tau).
\end{equation}
\end{proposition}

\begin{proof}
 We consider the case $\tau = \operatorname{id}$ only; the arguments made for $\pi$ then apply to $\tau$ analogously. Also, it suffices to consider a transposition $\pi=(k\;\,k\!+\!1)$ due to the representation property of $S^\pi$, see Sec.~\ref{sec:scatter}.
We rewrite \eqref{eq:fmndef} as
\begin{equation}\label{eq:cmesum}
\cme{m,n}{A}(\thetav,\etav )=\sum_{C\in \ccal_{m,n}} T_{C}(\thetav ,\etav ), \quad \text{where}\;
 T_{C}(\thetav ,\etav ):= (-1)^{|C|} \delta_C \, S_C \hscalar{ \lvector{C}{\thetav} }{ A \,\rvector{C}{\etav} }.
\end{equation}
We want to compute
\begin{equation}\label{Tpi}
 T_{C}(\thetav^{\pi},\etav)= (-1)^{|C|} \delta_C(\thetav^{\pi},\etav) \, S^{\sigma}(\thetav^{\pi}) \,
 S^{\rho}(\etav) \, \hscalar{ \lvector{C}{\thetav^\pi} }{ A \,\rvector{C}{\etav} },
\end{equation}
where $\thetav ^{\pi}=(\theta_{1},\ldots,\theta_{k+1},\theta_{k},\ldots,\theta_{m})$, and where $\sigma,\rho$ are permutations corresponding to $C$ by Lemma~\ref{lemma:scpermute}.
Let us distinguish the four cases where the indices $k$ and $k+1$ are each either contracted or non-contracted in $C$.

\begin{enumerate}[(a)]

 \item Neither $k$ nor $k+1$ are contracted. Since $\delta_{C}$ depends only on the contracted variables, we have  $\delta_{C}(\thetav^{\pi},\etav)=\delta_{C}(\thetav,\etav)$. Further, $\hscalar{ \lvector{C}{\thetav^\pi} }{ A \,\rvector{C}{\etav} } = S(\theta_{k}-\theta_{k+1}) \hscalar{ \lvector{C}{\thetav} }{ A \,\rvector{C}{\etav} }$ due to the exchange relations of the Zamolodchikov algebra.
 As for the factor $S^\sigma(\thetav^\pi)$:
 We have for all $1\leq j\leq |C|$ that either $k,k+1> l_{j}$ or $k,k+1< l_{j}$. It then follows from \eqref{eq:scpermutations}, \eqref{eq:Sperm} that $S^{\sigma}(\thetav^{\pi})=S^{\sigma}(\thetav)$.
 In total, we obtain
\begin{equation}\label{eq:tcpi-a}
S^\pi(\thetav) T_{C}(\thetav^{\pi},\etav)= T_{C}(\thetav,\etav).
\end{equation}

\item Both $k$ and $k+1$ are contracted. With $C=(m,n,\{(l_{1},r_{1}),\ldots,(k,r),(k+1,r'),\ldots,(l_{|C|},r_{|C|})\})$, let $C':=(m,n,\{(l_{1},r_{1}),\ldots,(k,r'),(k+1,r),\ldots,(l_{|C|},r_{|C|})\})$. Then $\delta_C(\thetav^\pi,\etav) = \delta_{C'}(\thetav,\etav)$. Also, $\lvector{C}{\thetav^\pi}=\lvector{C}{\thetav}=\lvector{C'}{\thetav}$, since $\lvector{C}{\cdotarg}$ does not depend on the contracted variables. Regarding the $S$-factors, we write using Eq.~\eqref{eq:Scompose},
    \begin{equation} \label{eq:spifactor}
    S^{\sigma}(\thetav^{\pi}) = S^{\pi\sigma}(\thetav) S^{\pi}(\thetav)^{-1}.
    \end{equation}
    One finds that $\pi \circ \sigma$ corresponds to $C'$ above in the sense of Eq.~\eqref{Scontrperm}, with the same $\rho$ for both $C$ and $C'$. Combining all this into \eqref{Tpi}, we obtain
    \begin{equation}\label{eq:tcpi-b}
    S^\pi(\thetav) T_{C}(\thetav^{\pi},\etav)=T_{C'}(\thetav,\etav).
    \end{equation}
    Note that the contraction $C'$ is again of type (b).

    \item $k$ is contracted, but $k+1$ is not.
    With $C=( m,n,\{(l_{1},r_{1}) \ldots (k,r) \ldots (l_{|C|},r_{|C|})\})$, let $C':=(m,n,\{(l_{1},r_{1})\ldots (k+1,r)\ldots (l_{|C|},r_{|C|}) \})$. We have $\delta_C(\thetav^\pi,\etav) = \delta_{C'}(\thetav,\etav)$ and $\lvector{C}{\thetav^\pi}=\lvector{C'}{\thetav}$. Moreover, as in \eqref{eq:spifactor}, $S^{\sigma}(\thetav^{\pi}) = S^{\pi \sigma}(\thetav) S^\pi(\thetav)^{-1}  $ where the permutation $\pi\circ \sigma$ corresponds to $C'$. Combining all in \eqref{Tpi}, we arrive at
    \begin{equation}\label{eq:tcpi-c}
    S^\pi(\thetav) T_{C}(\thetav^{\pi},\etav)= T_{C'}(\thetav,\etav).
    \end{equation}
  \item  $k+1$ is contracted, but $k$ is not. This case is fully analogous to (c); in fact, the contraction $C'$ in (c) is precisely of type (d).

\end{enumerate}

Summing over all contractions $C$ in \eqref{eq:cmesum}, we obtain from \eqref{eq:tcpi-a}, \eqref{eq:tcpi-b}, \eqref{eq:tcpi-c} that\\ $S^\pi(\thetav)\cme{m,n}{A}(\thetav^\pi,\etav ) =  \cme{m,n}{A}(\thetav,\etav )$ as claimed.
\end{proof}

In Eq.~\eqref{eq:fmndef}, we defined $\cme{m,n}{A}$ as a certain sum over matrix elements of $A$. We will now show that this formula can be inverted in the following sense.

\begin{proposition}\label{proposition:fmninversion}
For any $A \in \qf^{\omega}$,
\begin{equation}\label{eq:fmninverse}
\hscalar{ \lvector{}{\thetav} }{ A\, \rvector{}{\etav} }
= \sum_{C\in \ccal_{m,n}} \delta_C \,S_C(\thetav,\etav)  \, f_{m-|C|,n-|C|}^{[A]}( \hat \thetav, \hat \etav).
\end{equation}
\end{proposition}

\noindent (Here $\hat\thetav,\hat\etav$ denotes the variables obtained from $\thetav,\etav$ by dropping those components which are contracted in $C$.)

\begin{proof}
Inserting \eqref{eq:fmndef} into the right-hand side of \eqref{eq:fmninverse}, we need to show that
\begin{equation}\label{insert}
\hscalar{ \lvector{}{\thetav} }{A  \rvector{}{\etav} }
 =\sum_{C\in\ccal_{m,n}}\delta_{C}S_C(\thetav,\etav )
 \;\;\sum_{\mathclap{C'\in\ccal_{m-|C|,n-|C|}}}\;\; (-1)^{|C'|}\delta_{C'}S_{C'}(\hat{\thetav },\hat\etav)
 \hscalar{ \lvector{C \dot\cup C'}{\thetav} }{A \rvector{C \dot\cup C'}{\etav} }.
\end{equation}
Using Lemma~\ref{lemma:contractcompose}, we find
\begin{equation}
\rhs{insert} = \;\; \sum_{\mathclap{\substack{C\in\ccal_{m,n} \\ C'\in\ccal_{m-|C|,n-|C|}}}} \;\;
(-1)^{|C'|}\delta_{C\dot{\cup}C'}S_{C\dot{\cup}C'}(\thetav,\etav )\; \hscalar{ \lvector{C \dot\cup C'}{\thetav} }{A \rvector{C \dot\cup C'}{\etav} }.
\end{equation}
Denoting $D:=C\dot{\cup}C'$, we can reorganize the sum over $C,C'$ as follows:
\begin{equation}\label{eq:dsum}
\rhs{insert}=\sum_{D\in\ccal_{m,n}}\Big(\sum_{D=C\dot{\cup}C'}(-1)^{|C'|}\Big)\delta_{D}S_{D}
\hscalar{ \lvector{D}{\thetav} }{A \rvector{D}{\etav} }.
\end{equation}
After computing the inner sum in the above formula (at fixed $D$),
\begin{equation}
\sum_{D=C\dot{\cup}C'}(-1)^{|C'|}=\sum_{j=0}^{|D|}(-1)^{j}\binom{|D|}{j}=\begin{cases} 0 \quad&\mathrm{if}\quad |D|\geq 1,\\ 1 &\mathrm{if}\quad |D|=0,\end{cases}
\end{equation}
we find that the right hand side of \eqref{eq:dsum} gives $\hscalar{ \lvector{}{\thetav} }{ A\, \rvector{}{\etav} }$ as claimed.
\end{proof}

\subsection{Existence and uniqueness of the expansion}

After these preparations, we can now go on to establish the expansion \eqref{eq:expansionz}. We first note that the individual expansion terms of the form $z^{\dagger m}z^n(g)$ are well-defined elements of $\qf^\omega$ if $\onorm{g}{m \times n}<\infty$; see Prop.~\ref{proposition:zzdcrossnorm}. In particular, we can compute their fully contracted matrix elements $\cme{m,n}{\cdotarg}$. 
We will now show that the $z^{\dagger m}z^n(g)$ form a ``dual basis'' to the $\cme{m,n}{\cdotarg}$ in the following sense.

\begin{lemma}\label{lemma:fmnbasis}
In the sense of distributions, there holds:
\begin{equation}\label{basisprop}
\cmelong{m,n}{z^{\dagger m'}(\thetav')z^{n'}(\etav')} (\thetav,\etav)
= m!n! \delta_{m,m'}\delta_{n,n'}
\operatorname{Sym}_{S,\thetav} \delta^m(\thetav-\thetav')
\operatorname{Sym}_{S,\etav}\delta^n(\etav-\etav').
\end{equation}

\end{lemma}

\begin{proof}
Let us set $A:=z^{\dagger m'}(\thetav ')z^{n'}(\etav ')$. If $m-m'\neq n-n'$, then in \eqref{eq:fmndef} all matrix elements of $A$ vanish, hence $\cme{m,n}{A} =0$ and the claim follows. Therefore, in the following let $k:=m-m'=n-n'$.
If $k<0$, then by \eqref{eq:fmndef} we have $\cme{m,n}{A} =0$ and again the claim follows.
If $k=0$, then
\begin{multline}\label{eq:k0case}
\cme{m,n}{A}(\thetav,\etav)
=\hscalar{ z^{\dagger m}(\thetav )\Omega }{ z^{\dagger m}(\thetav') z^{n}(\etav') J z^{\dagger n}(\etav) \Omega}
\\
=\hscalar{ z^{\dagger m}(\thetav )\Omega }{ z^{\dagger m}(\thetav ')\Omega} \hscalar{ \Omega}{z^{n}(\etav ')J z^{\dagger n}(\etav )\Omega}
=m!n!\operatorname{Sym}_{S,\thetav} \delta^m(\thetav-\thetav')\operatorname{Sym}_{S,\etav}\delta^ n(\etav-\etav'),
\end{multline}
as claimed. For $k>0$, we use induction on $k$. From Prop.~\ref{proposition:fmninversion}, we have
\begin{multline}\label{eq:mtxk}
\hscalar{ \lvector{}{\thetav} }{ A \rvector{}{\etav}}
=\sum_{C\in \ccal_{m,n}}\delta_{C}S_Cf_{m-|C|,n-|C|}^{[A]}(\hat{\thetav },\hat{\etav })\\
=\cme{m,n}{A}(\thetav ,\etav )
  +\sum_{\substack{C\in \ccal_{m,n}  \\ |C|=k}}\delta_{C}S_C
   m'!n'!\operatorname{Sym}_{S,\hat{\thetav}} \delta^{m'}(\hat{\thetav }-\thetav ')\operatorname{Sym}_{S,\hat{\etav}} \delta^{ n'}(\hat{\etav }-\etav ').
\end{multline}
(Here we have used that by induction hypothesis, $\cme{m-|C|,n-|C|}{A}=0$ for $|C|\geq 1$, except for $|C|=k$, where \eqref{eq:k0case} applies.)
In view of \eqref{symstheta}, it therefore suffices to show that
\begin{equation}\label{matrixd}
\hscalar{ \lvector{}{\thetav} }{ A \rvector{}{\etav}}
=
m'!n'! \operatorname{Sym}_{S^{-1},\thetav'} \operatorname{Sym}_{S^{-1},\etav'}
 \sum_{\substack{C\in \ccal_{m,n}  \\ |C|=k}}\delta_{C}S_C \delta^{m'}(\hat{\thetav }-\thetav ') \delta^{n'}(\hat{\etav }-\etav ').
\end{equation}
We rewrite $S_C$ as $S^\sigma S^\rho$, where $\sigma$ and $\rho$ are the permutations of Eq.~\eqref{Scontrperm}, depending on $C$. Rewriting the arguments of the delta functions in terms of $\sigma$, $\rho$ as well, we obtain
\begin{multline}
\rhs{matrixd}
=m'!n'! \operatorname{Sym}_{S^{-1},\thetav'} \operatorname{Sym}_{S^{-1},\etav'}
\sum_{\substack{C\in \ccal_{m,n}  \\ |C|=k}}
 S^\sigma(\thetav)S^\rho(\etav)
\\
\times
\Big(\prod_{j=1}^{k}\delta(\theta_{\sigma(m-j+1)}-\eta_{\rho(j)}) \Big)
\Big(\prod_{j=1}^{m'}\delta(\theta_{\sigma(j)}-\theta_{j}') \Big)
\Big(\prod_{j=1}^{n'}\delta(\eta_{\rho(j+k)}-\eta_{j}') \Big).
\end{multline}
Since both $\hscalar{\lvector{}{\thetav}}{A\rvector{}{\etav}}$ and $\cme{m,n}{A}$ are $S$-symmetric in the variables $\thetav,\etav$, we know from \eqref{eq:mtxk} that the right hand side of \eqref{matrixd} must be $S$-symmetric too; we can therefore take the $S$-symmetric part of each term in the sum. Then, using the formula $\operatorname{Sym}_{S,\thetav}(S^{\sigma}(\thetav)g(\thetav^{\sigma}))=\operatorname{Sym}_{S,\thetav} g(\thetav)$ to simplify the expression, we find that the terms in the sum do not actually depend on $C$; the sum, which contains $\binom{n}{k}\binom{m}{k}k!$ terms, can then be computed. Due to \eqref{symstheta} the symmetrization in $\thetav',\etav'$ can be dropped as well in favor of the symmetrization in $\thetav,\etav$. After working out the numerical factors, one arrives at
\begin{multline}
\rhs{matrixd}
=\frac{m!n!}{k!}
\operatorname{Sym}_{S,\thetav}\operatorname{Sym}_{S,\etav}
\Big(\prod_{j=1}^{k}\delta(\theta_{m-j+1}-\eta_{j}) \Big)
\Big(\prod_{j=1}^{m'}\delta(\theta_{j}-\theta_{j}') \Big)
\Big(\prod_{j=1}^{n'}\delta(\eta_{j+k}-\eta_{j}') \Big).
\end{multline}
However, the right hand side is just $\hscalar{ \lvector{}{\thetav} }{ z^{\dagger m'}(\thetav ')z^{n'}(\etav ') \rvector{}{\etav}}$ -- this is seen by a straightforward computation using the Zamolodchikov relations. That shows \eqref{matrixd} and therefore concludes the proof.
\end{proof}

In the next step, we show that any sequence of distributions $g_{mn}$ with finite norms $\onorm{g_{mn}}{m \times n}$ \emph{defines} a quadratic form by means of the series \eqref{eq:expansionz}, and that the $g_{mn}$ can be recovered from this quadratic form by computing fully contracted matrix elements. This amounts to a uniqueness result for the series expansion. We stress that the series is actually a finite sum in the matrix elements that we consider, therefore convergence issues do not arise.

\begin{proposition}\label{proposition:expansionunique}
For any $m,n\in\nbb_0$, let $g_{mn}\in\dcal(\rbb^{m+n})'$ with $\gnorm{g_{mn}}{m \times n}^{\omega} < \infty$. Then,
\begin{equation}\label{arakig}
A := \sum_{m,n=0}^\infty \int \frac{d\thetav  \,d \etav}{m!n!} \, g_{mn}(\thetav,\etav) \,z^{\dagger m}(\thetav)z^n(\etav)
\end{equation}
defines an element of $\mathcal{Q}^{\omega}$, and $\cme{m,n}{A}(\thetav,\etav) = \operatorname{Sym}_{S,\thetav} \operatorname{Sym}_{S,\etav} g_{mn}(\thetav,\etav)$.
\end{proposition}

\begin{proof}
We only need to consider \eqref{arakig} evaluated between vectors of finite particle number, where the sum on the right hand side of \eqref{arakig} is finite. By Prop.~\ref{proposition:zzdcrossnorm}, every summand -- and hence the sum -- is a well-defined quadratic form in $\mathcal{Q}^{\omega}$. We now show the proposed formula for $\cme{m,n}{A}$.
\begin{multline}
\cme{m,n}{A}(\thetav ,\etav )=\sum_{m',n'= 0}^\infty \int \frac{d\thetav'd\etav'}{m'!n'!}g_{m'n'}(\thetav ',\etav ')
\cmelong{m,n}{z^{\dagger m'}(\thetav ')z^{n'}(\etav ')}(\thetav,\etav)\\
=\sum_{m',n'=0}^\infty \int \frac{d\thetav'd\etav'}{m'!n'!}g_{m'n'}(\thetav ',\etav ')m!n! \delta_{m,m'}\delta_{n,n'}
\operatorname{Sym}_{S,\thetav} \delta^{m}(\thetav-\thetav')
\operatorname{Sym}_{S,\etav} \delta^{n}(\etav-\etav')\\
=\operatorname{Sym}_{S,\thetav} \operatorname{Sym}_{S,\etav} g_{mn}(\thetav ,\etav ),
\end{multline}
where in the second equality we made use of Lemma~\ref{lemma:fmnbasis}.
\end{proof}

Finally, we are in the position to show that any $A \in \qf^{\omega}$ can be expanded into a series as in \eqref{eq:expansionz}, using the $\cme{m,n}{A}$ as expansion coefficients.
\begin{theorem}\label{theorem:arakiexpansion}
If $A \in \mathcal{Q}^{\omega}$, then in the sense of quadratic forms,
\begin{equation}\label{eq:arakiexp}
A = \sum_{m,n=0}^\infty \int \frac{d \thetav \, d \etav}{m!n!} \,\cme{m,n}{A} (\thetav,\etav)
z^{\dagger m}(\thetav)z^ n(\etav).
\end{equation}
\end{theorem}

\begin{proof}
According to Prop.~\ref{proposition:fmnbound}, we have $\gnorm{\cme{m,n}{A}}{m\times n}^{\omega}<\infty$, thus Prop.~\ref{proposition:expansionunique} shows that the right-hand side of \eqref{eq:arakiexp} exists in $\qf^{\omega}$. To establish equality in \eqref{eq:arakiexp}, we need to show that both sides agree in all matrix elements. In view of Prop.~\ref{proposition:fmninversion}, it suffices to show that they agree in all $\cme{m,n}{\cdotarg}$. But this is the case by Prop.~\ref{proposition:expansionunique}, with $g_{mn}=\cme{m,n}{A}$. (We have used here that the $\cme{m,n}{A}$ are $S$-symmetric by Prop.~\ref{proposition:fmnsymm}.)
\end{proof}

\subsection{Behavior under symmetry transformations}\label{sec:fsymtrans}

We will now investigate how the expansion coefficients $\cme{m,n}{A}$ change when a symmetry transformation acts on $A$. In the case of translations and boosts, this is easy to describe.
\begin{proposition}\label{proposition:fmnpoincare}
For any $A \in \mathcal{Q}^{\omega}$, $x \in \rbb^2$, and $\lambda \in \rbb$,
\begin{equation}\label{eq:fmnpoincare}
\cmelong{m,n}{U(x,\lambda)AU(x,\lambda)^\ast}(\thetav,\etav) = e^{i(p(\thetav)-p(\etav))\cdot x}
 \cme{m,n}{A}(\thetav-\lambdav,\etav-\lambdav),
\end{equation}
where $\thetav-\lambdav = (\theta_1-\lambda, \ldots, \theta_m-\lambda)$, similarly for $\etav-\lambdav$.
\end{proposition}

\begin{proof}
From the definition \eqref{eq:fmndef}, we know that
\begin{equation}
\cmelong{m,n}{U(x,\lambda)AU(x,\lambda)^{*}}(\thetav ,\etav )=\sum_{C\in \ccal_{m,n}} (-1)^{|C|} \delta_C \, S_C \,
\hscalar{ \lvector{C}{\thetav}}{ U(x,\lambda)A U(x,\lambda)^{*} \rvector{C}{\etav} }.
\end{equation}
Here Eq.~\eqref{eq:translboostdef} yields $U(x,\lambda)^\ast\rvector{C}{\etav} = \exp(-i \sum_{k \not\in \{r_j\}} p(\eta_{k-m})\cdot x) \,\rvector{C}{\etav-\lambdav}$, similarly for $\lvector{C}{\thetav}$. In view of the support of $\delta_C$, and since the factors $\delta_C$, $S_C$ depend on differences of rapidities only, the result follows.
\end{proof}

The behavior of the coefficients under space-time reflections (antiunitarily represented by $J$) is more involved.
To describe it, we introduce for any contraction $C=(m,n,\{(\ell_j,r_j)\})$ the ``reflected`` contraction
$C^J=(n,m,\{(r_j-m,\ell_j+n)\})$, i.e., the one that is obtained from $C$ by swapping the roles of left and right indices.
We also introduce the following factor associated with $C$:
\begin{equation}\label{eq:rcdef}
 R_C(\thetav,\etav) := \prod_{j=1}^{|C|} \Big(1-\prod_{p_j=1}^{m+n} S^{(m)}_{l_j,p_j}(\thetav,\etav) \Big).
\end{equation}
This factor encodes, in a sense, the interaction of the model; note that in the free case $S=1$, one has $R_C = \delta_{|C|,0}$.
Its most important mathematical property is as follows.

\begin{lemma}\label{lemma:RSJ}
For any contraction $C$, we have
\begin{equation}\label{eq:RSJ}
\delta_{C^{J}}(\thetav,\etav)S_{C^{J}}(\thetav,\etav)R_{C^{J}}(\thetav,\etav)
=(-1)^{|C|} \delta_{C}(\etav,\thetav)S_{C}(\etav,\thetav)R_{C}(\etav,\thetav).
\end{equation}
\end{lemma}

\begin{proof}
Using Lemma~\ref{lemma:scpermute}, we rewrite \eqref{eq:RSJ} equivalently as
\begin{equation}\label{eq:RSJperm}
\delta_{C^{J}}(\thetav,\etav)S^{\sigma'}(\thetav) S^{\rho'}(\etav)R_{C^{J}}(\thetav,\etav)
=(-1)^{|C|}\delta_{C}(\etav,\thetav) S^\sigma(\etav) S^\rho(\thetav) R_{C}(\etav,\thetav),
\end{equation}
where $\sigma,\rho$ and $\sigma',\rho'$ correspond to $C$ and $C^J$, respectively. A short computation shows that we can choose $\rho' = \sigma \circ \pi$ with the permutation
\begin{equation}
\pi=\begin{pmatrix}1 && \ldots \hat{l} \ldots && m & l_1 & \ldots & l_{|C|}
\\ l_{|C|} & \ldots & l_1 & 1 && \ldots \hat{l} \ldots && m \end{pmatrix},
\end{equation}
and therefore $S^{\rho'}(\etav) = S^\sigma(\etav)S^{\pi}(\etav^\sigma) $ due to \eqref{eq:Scompose}. Correspondingly, one finds $S^{\sigma'}(\thetav) = S^\rho(\thetav)  S^{\tau}(\thetav^\rho) $ with an analogously defined permutation $\tau$. From the definition \eqref{eq:Sperm}, we explicitly compute on the support of $\delta_{C^J}$,
\begin{equation}\label{eq:spitau}
S^{\pi}(\etav^{\sigma}) S^{\tau}(\thetav^{\rho}) = \prod_{j=1}^{|C|}\prod_{p_{j}=1}^{m+n}S^{(m)}_{l_{j},p_{j}}(\etav,\thetav).
\end{equation}
From \eqref{eq:rcdef} and \eqref{eq:spitau}, one can now see that, again on the support of $\delta_{C^J}$,
\begin{equation}
  R_{C^J}(\thetav,\etav) S^{\pi}(\etav^{\sigma}) S^{\tau}(\thetav^{\rho}) = (-1)^{|C|} R_{C}(\etav,\thetav);
\end{equation}
this shows \eqref{eq:RSJperm}.
\end{proof}

We will now see how the factor $R_C$ describes the action of space-time reflections on the level of expansion coefficients.
\begin{proposition}\label{proposition:fmnreflected}
For any $A \in \qf^{\omega}$,
\begin{equation} \label{eq:fmnreflected}
\cme{m,n}{J A^\ast J}(\thetav,\etav) = \sum_{C \in \ccal_{m,n}}
(-1)^{|C|} \delta_C S_C
R_C(\thetav,\etav)
\cme{n-|C|,m-|C|}{A}(\hat \etav ,\hat\thetav).
\end{equation}
\end{proposition}

\begin{proof}
We first note that for any contraction $C\in\ccal_{m,n}$, we have $J \lvector{C}{\cdotarg} = \rvector{C^J}{\cdotarg}$ and $J \rvector{C}{\cdotarg} = \lvector{C^J}{\cdotarg}$. By replacing $A$ with $JA^{*}J$ in the definition of $\cme{m,n}{A}$, Eq.~\eqref{eq:fmndef}, we then obtain:
\begin{equation}
\cme{m,n}{JA^{*}J}(\thetav,\etav)=
\sum_{C\in\ccal_{m,n}}(-1)^{|C|}\delta_{C}S_{C}(\thetav,\etav)
\hscalar{ \lvector{C^J}{\etav} }{ A \, \rvector{C^J}{\thetav}}.
\end{equation}
Using Prop.~\ref{proposition:fmninversion} in the formula above, we find
\begin{equation}\label{maineq}
\cme{m,n}{JA^{*}J}(\thetav,\etav)=
\;\;\sum_{\mathclap{\substack{C\in\ccal_{m,n} \\ C'\in\ccal_{n-|C|,m-|C|} }}}\;\;
(-1)^{|C|}\delta_{C} S_{C}(\thetav,\etav ) \delta_{C'} S_{C'}(\hat{\etav},\hat{\thetav}) \cme{n-|C|-|C'|,m-|C|-|C'|}{A}(\Hat{\Hat{\etav}},\Hat{\Hat{\thetav}}),
\end{equation}
where in $\hat{\thetav},\hat\etav$ variables are ``dropped'' with respect to $C$, but in $\Hat{\Hat{\thetav}},\Hat{\Hat\etav}$ with respect to $C \dot\cup C^{\prime J}$.
We apply Lemma~\ref{lemma:contractcompose} (with $C^{\prime J}$ in place of $C'$), and reorganize the sum, setting $D:=C \dot\cup C^{\prime J}$. This yields
\begin{equation}\label{eq:dsumJ}
\cme{m,n}{JA^{*}J}(\thetav,\etav)=\sum_{D\in\ccal_{m,n}}(-1)^{|D|} \delta_{D} S_{D}(\thetav,\etav)
\underbrace{\Big(\sum_{D = C \dot\cup C^{\prime J}} (-1)^{|C'|} \frac{S_{C^{\prime}}(\hat\etav,\hat\thetav)}{S_{C^{\prime J}}(\hat\thetav,\hat\etav)} \Big)}_{(\ast)}
\cme{n-|D|,m-|D|}{A}(\hat{\hat{\etav}},\hat{\hat{\thetav}}).
\end{equation}
It remains to compute the inner sum $(\ast)$.
Using Lemma~\ref{lemma:RSJ}, we have on the support of $\delta_D$,
\begin{equation}\label{eq:dsumrewritten}
 (\ast) = \sum_{D = C \dot\cup C^{\prime J}} \frac{R_{C^{\prime J}}(\hat\thetav,\hat\etav)}{R_{C^{\prime}}(\hat\etav,\hat\thetav)}
  = \sum_{D = C \dot\cup C^{\prime J}}  \prod_{j \in \{r_i'-n\}} \frac{1-a_j}{1-a_j^{-1}},
\quad \text{where} \; a_j:=\prod_{p=1}^{m+n} S_{j,p}^{(m)}(\thetav,\etav).
\end{equation}
Using the distributive law, we find
\begin{equation}
 (\ast) =  \prod_{j \in \{\ell_i\} \cup \{r_i'-n\}} \Big(1 + \frac{1-a_j}{1-a_j^{-1}}\Big)
=  \prod_{j \in \{\ell_i\} \cup \{r_i'-n\}} (1-a_j)
 = R_D(\thetav,\etav).
\end{equation}
Inserting this result into \eqref{eq:dsumJ} concludes the proof.
\end{proof}

\section{Warped convolutions}\label{sec:warped}

We will now investigate whether the expansion coefficients $\cme{m,n}{A}$ can be expressed in a simpler way than in their definition \eqref{eq:fmndef}, and thus, are amenable to a more natural interpretation.

This is certainly possible in the case $S=1$, that is, in free field theory, where the Zamolodchikov operators $\zd$ and $z$ are the usual Bose annihilation and creation operators $\ad$ and $a$. As remarked in the introduction, we can in this case write the expansion coefficients as
\begin{equation}\label{eq:fmnfree}
 \cme{m,n}{A}(\thetav,\etav) = \bighscalar{\Omega}{ [a(\theta_m), \ldots [a(\theta_1),[\ldots[A,\ad(\eta_n)]\ldots ,\ad(\eta_1)] \ldots ] \;\Omega  }.
\end{equation}
This can be verified in a straightforward manner in the case $A=a^{\dagger m'} a^{n'}(f)$, using the rules of the CCR. It then holds for all quadratic forms $A$ by linearity, making use of the expansion established in Thm.~\ref{theorem:arakiexpansion}.

A formula of this kind can certainly be extended to some other situations as well. For example, in the Ising model ($S=-1$), where $\zd,z$ fulfill the CAR, we can define a graded commutator $[\cdotarg, \cdotarg]_g$, equaling the commutator between even operators and the anticommutator between odd operators (with respect to the adjoint action of $(-1)^N$). One then obtains
\begin{equation}\label{eq:fmnising}
 \cme{m,n}{A}(\thetav,\etav) = \bighscalar{\Omega}{ [z(\theta_m), \ldots [z(\theta_1),[\ldots[A,\zd(\eta_n)]_g\ldots ,\zd(\eta_1)]_g \ldots ]_g \;\Omega  },
\end{equation}
in full analogy with \eqref{eq:fmnfree}. Again, a proof would follow the idea of explicitly computing the commutators in the case $A=z^{\dagger m'} z^{n'}(f)$.

It is natural to ask whether this kind of nested commutator expression can be generalized to a larger class of models. Here we will focus on those obtained from the \emph{warped convolution} construction described by Buchholz, Summers and Lechner \cite{BuchholzSummers:2008,BuchholzSummersLechner:2011}. In this approach, one starts from a given quantum field theory and deforms the algebras of observables, thus constructing a new theory. This takes a skew symmetric matrix $Q$ as a deformation parameter; it is equivalent to a Rieffel deformation \cite{Rieffel:1993} with respect to the action of the translation group, and can alternatively be interpreted in terms of a quantum field theory on noncommutative space-time \cite{GrosseLechner:2007}. While the intent of \cite{BuchholzSummersLechner:2011} was to apply this deformation to a general, possibly interacting quantum field theory, in particular in 2+1 and more space-time dimensions, we return here to 1+1 dimensional free field theory as the starting point. The resulting deformed theory is then known to be equivalent to an integrable model with a certain simple type of scattering function $S$; cf.~\cite{GrosseLechner:2007}. We shall see this equivalence explicitly below.

Our aim is to define, for this particular class of models, a ``deformed commutator'' $[\cdotarg, \cdotarg]_Q$ which makes an analogue of the formula \eqref{eq:fmnfree} hold.

Let us first mention some technical preliminaries. In all what follows, $\hcal$ and related spaces will denote the objects associated with the free field, $S=1$. Also, we only consider the indicatrix $\omega=0$ and drop the superscript $\omega$ from all our objects.

With $\ccal^\infty$, we denote the subalgebra of $\boundedops$ consisting of norm-smooth operators with respect to the adjoint action of translations; we equip $\ccal^\infty$ with the usual Fr\'echet topology. Correspondingly, let $\qf^\infty \subset \qf$ be the space of quadratic forms $A$ that fulfill $\fpnp_k A \fpnp_k \in \ccal^\infty$ for all $k$. We also consider the following subspace $\fcal^\infty$ of $\qf^\infty$: For $A \in \fcal^\infty$ and any $k \in \nbb$, there exists $k' \in \nbb$ so that $\fpnp_{k'} A \fpnp_k=A \fpnp_k$ and  $\fpnp_{k} A \fpnp_{k'}=\fpnp_{k} A$. Then $\qf^\infty$ is a bimodule over $\fcal^\infty$. Note that $\fcal^\infty \subset \qf^\infty$, $\ccal^\infty\subset\qf^\infty$, but $\ccal^\infty \not \subset \fcal^\infty$.

Further, we consider $\fcal^\infty$-valued distributions on $\rbb^m$: By this we mean linear maps $\dcal(\rbb^m) \to \fcal^\infty$, $f \mapsto A(f)$, such that for any $k$, the number $k'$ above can be chosen independent of $f$, and such that the maps $f \mapsto A(f) \fpnp_k$ and $f \mapsto \fpnp_k  A(f)$ are continuous in the Fr\'echet topology of $\ccal^\infty$. Products of $\fcal^\infty$-valued distributions in independent variables are again $\fcal^\infty$-valued distributions. As before, we will usually write distributions in terms of their formal kernels, $A(f) = \int A(\thetav) f(\thetav) d\thetav$.

Since the (anti)unitary representation $U$ of the proper Poincar\'e group commutes with the particle number operator, the symmetry group acts, via the adjoint action of $U$, on $\ccal^\infty$, $\qf^\infty$, $\fcal^\infty$, and on $\fcal^\infty$-valued distributions. Here we are particularly interested in the action of translations operators $U(x):=U(x,0)$. We say that an $\fcal^\infty$-valued distribution $A$ is \emph{homogeneous} if there is a smooth function $\varphi_A:\rbb^m \to \rbb^2$ such that
\begin{equation}\label{homogeneous}
\forall x \in \rbb^2: \quad U(x)A(\thetav)U(x)\st=e^{i\varphi_A(\thetav)\cdot x}A(\thetav).
\end{equation}
We call $\varphi_A$ the \emph{momentum transfer} of $A$. If $A(\thetav),B(\etav)$ are both homogeneous, then so is $A(\thetav)B(\etav)$, with momentum transfer $\varphi_{AB}(\thetav,\etav)=\varphi_A(\thetav)+\varphi_B(\etav)$.
Important examples of homogeneous distributions are $\ad(\theta)$, $a(\eta)$, and $a^{\dagger m}a^{n}(\thetav,\etav)$, with momentum transfer $p(\theta)$, $-p(\eta)$, and $p(\thetav)-p(\etav)$, respectively, as well as their deformed versions to be considered below.

We now proceed to the warped convolution. We denote the (joint) spectral measure of the momentum operator as $dE(p)$. Further, we fix a $2\times 2$ matrix $Q$ which is skew symmetric with respect to the Minkowski scalar product, i.e., $xQy:=x\cdot (Qy) = -(Qx) \cdot y$. The warped convolution $\tau_Q(A)$ of an operator $A$ is formally defined by
\begin{equation}\label{eq:warpedintegral}
\tau_{Q}(A) := \int U(Qp) A U(Qp)\st\,dE(p)=\int dE(p) \,U(Qp) A U(Qp)\st.
\end{equation}
This integral must be taken with care, since the integrand has constant norm. However, Buchholz, Lechner and Summers \cite{BuchholzSummersLechner:2011} were able to define it for smooth operators $A$ in the sense of an oscillatory integral, yielding a bijective map $\tau_Q:\ccal^\infty \to \ccal^\infty$. We need to extend this map to our space of quadratic forms, and will use the projectors $\fpnp_k$ to that end.\footnote{%
An alternative approach would be to use a generalization of the deformation integral to locally convex spaces, as established in \cite{LechnerWaldmann:2011}.} 
Since the $\fpnp_k$ commute with $U(x)$, we can establish for $A \in \ccal^\infty$,
\begin{equation}\label{eq:qkcommute}
   \tau_Q(A \fpnp_k) = \tau_Q(A) \fpnp_k, \quad \tau_Q(\fpnp_k A ) = \fpnp_k \tau_Q(A).
\end{equation}
This allows us to  extend $\tau_Q$ to quadratic forms $A \in \qcal^\infty$: We set for $\psi,\chi\in\fpn$,
\begin{equation}\label{eq:qkcompat}
   \hscalar{\psi}{ \tau_Q(A) \chi} := \hscalar{\psi}{ \tau_Q(\fpnp_k A \fpnp_k) \chi}
\end{equation}
with $k$ chosen large enough for $\psi,\chi$ so that, by \eqref{eq:qkcommute}, the expression on the r.h.s.~becomes independent of $k$. The relations \eqref{eq:qkcommute} then hold for all $A \in \qcal^\infty$. We state the most important properties of the map $\tau_Q$.

\begin{proposition}\label{proposition:tauq}
  For any skew symmetric matrices $Q,Q'$, we have:

\begin{enumerate}
  \renewcommand{\theenumi}{(\roman{enumi})}
  \renewcommand{\labelenumi}{\theenumi}
  \item\label{it:warpedcontinuous}
    $\tau_Q: \ccal^\infty \to \ccal^\infty$ is continuous.
  \item\label{it:warpedcompose}
    $\tau_Q \tau_{Q'} = \tau_{Q+Q'}$, $\tau_0 = \id$, $\tau_Q^{-1} = \tau_{-Q}$.
  \item\label{it:warpedtranslate}
    $\tau_Q(U(x) A) = U(x) \tau_Q(A)$, $\tau_Q(A U(x)) = \tau_Q(A)U(x)$ for any $x \in \rbb^2$ and $A \in \qf^\infty$.
  \item\label{it:warpediso}
    $\tau_Q: \ccal^\infty \to \ccal^\infty$, $\tau_Q:\fcal^\infty\to\fcal^\infty$, $\tau_Q:\qcal^\infty \to \qcal^\infty$ are $\ast$-preserving vector space isomorphisms.
  \item\label{it:warpeddistribution}
    If $A$ is an $\fcal^\infty$-valued distribution, then $\tau_Q(A) : f \mapsto \tau_Q(A(f))$ is an $\fcal^\infty$-valued distribution as well.
    If $A$ is homogeneous, then so is $\tau_Q(A)$, with the same momentum transfer as $A$.
\end{enumerate}
\end{proposition}

\begin{proof}
Part~\ref{it:warpedcontinuous} follows similar to \cite[Prop.~2.7(ii)]{BuchholzSummersLechner:2011}: With notation used there, the inclusion $\imath$ of $\ccal^\infty$, equipped with the Fr\'echet topology induced by $\gnorm{\cdotarg}{}$, into itself equipped with the Fr\'echet topology induced by $\gnorm{\cdotarg}{Q}$, is continuous \cite[Lemma~7.2]{Rieffel:1993}. The map $\pi_Q$ in turn is norm-preserving between $\gnorm{\cdotarg}{Q}$ and $\gnorm{\cdotarg}{}$ and hence intertwines the associated Fr\'echet topologies as well. In our terms, $\tau_Q = \pi_Q \circ \imath$.---%
 For~\ref{it:warpedcompose}, the relation $\tau_Q \tau_{Q'} = \tau_{Q+Q'}$ was shown on $\ccal^\infty$ in \cite[Prop.~2.11]{BuchholzSummersLechner:2011}; it extends to $\qcal^\infty$ by using \eqref{eq:qkcompat}. $\tau_0 = \id$ is immediate from the definition, and it follows that $\tau_Q\tau_{-Q}=\tau_{Q-Q}=\id$.---%
 Part~\ref{it:warpedtranslate} can be obtained by explicit computation, e.g., from \cite[Eq.~(2.4)]{BuchholzSummersLechner:2011}, noting that $U(x)$ commutes with $\fpnp_k$ for the case $A \in \qcal^\infty$.---%
 In \ref{it:warpediso}, the inclusion $\tau_Q(\ccal^\infty)\subset \ccal^\infty$ was already noted and  $\tau_Q(\fcal^\infty)\subset\fcal^\infty$, $\tau_Q(\qf^\infty)\subset\qf^\infty$ then follow from \eqref{eq:qkcommute}. The map $\tau_Q$ is invertible in each case by \ref{it:warpedcompose}. Linearity is clear. Also, $\tau_Q(A\st) = \tau_Q(A)\st$ for any $A \in \ccal^\infty$ by \cite[Lemma~2.2(ii)]{BuchholzSummersLechner:2011}; this extends to $\fcal^\infty$ and $\qf^\infty$ with the help of \eqref{eq:qkcommute}, in the sense of form adjoints.---%
 For \ref{it:warpeddistribution}, if $\fpnp_{k'} A \fpnp_{k} =  A \fpnp_{k}$ where $k'$ depends on $k$ but not on $f$, then $\fpnp_{k'} \tau_Q(A(f)) \fpnp_{k}= \tau_Q(A(f)) \fpnp_{k}$ by an application of \eqref{eq:qkcommute}, and similarly for $\fpnp_{k} A \fpnp_{k'}$. Since $f \mapsto \tau_Q(A(f)) \fpnp_k$ and $f \mapsto \fpnp_k \tau_Q(A(f))$ are continuous by \eqref{eq:qkcommute} and \ref{it:warpedcontinuous}, the map $f \mapsto  \tau_Q(A(f))$ is a well-defined $\fcal^\infty$-valued distribution. Homogeneity of $A$ implies homogeneity of $\tau_Q(A)$ by an application of \ref{it:warpedtranslate}.
\end{proof}

It should be remarked that $\tau_Q$ is a vector space isomorphism, but is \emph{not} multiplicative -- it deforms the operator product in this sense.

For the action of $\tau_Q$ on homogeneous distributions, we can obtain more explicit results. We start with a simple formula which is intuitively obvious from the integral representation \eqref{eq:warpedintegral}; the full proof however involves details of the oscillatory integral and requires some work.

\begin{lemma}\label{lemma:homogdeformvec}
Let $A$ be a homogeneous $\fcal^\infty$-valued distribution on $\rbb^m$. In the sense of distributional kernels, it holds that
 \begin{equation}\label{eq:QCvector}
    \tau_Q(A(\thetav)) \rvector{}{\etav} = e^{i \varphi_A(\thetav) Q p(\etav) } A(\thetav) \rvector{}{\etav}.
 \end{equation}
 
\end{lemma}

\begin{proof}
Let $f\in\dcal(\rbb^m)$, $g\in\dcal(\rbb^n)$. Choose $h_1,h_2\in\scal(\rbb^2)$ such that $h_1=1$ on a neighborhood of 0, $h_2(0)=1$, and such that the Fourier transform $\tilde h_2$ of $h_2$ has compact support. Since $\rvector{}{g}$ has fixed particle number and is smooth with respect to translations, we have by \cite[Eq.~(2.4)]{BuchholzSummersLechner:2011},
\begin{equation}
  \tau_Q(A(f)) \rvector{}{g} = (2\pi)^{-2}  \lim_{\epsilon \to 0} \iint dx\,dy\, h_1(\epsilon x) h_2(\epsilon y) e^{-ix \cdot y} U(Qx) A(f) U(Qx)^\ast U(y) \rvector{}{g}.
\end{equation}
Using homogeneity, this yields
\begin{equation}\label{eq:crlimit}
\begin{aligned}
  \tau_Q(A(f)) \rvector{}{g} 
 &= (2\pi)^{-2}\lim_{\epsilon \to 0}  \iint dx\,dy\, h_1(\epsilon x) h_2(\epsilon y) e^{-ix \cdot y} A(e^{i \varphi_A(\cdotarg) Qx}f) \rvector{}{e^{i p(\cdotarg) y}g} 
\\
 &= \lim_{\epsilon \to 0}A\rvector{}{F_\epsilon},
\end{aligned}
\end{equation}
where $F_\epsilon\in\dcal(\rbb^{m+n})$ is the test function
\begin{equation}
\begin{aligned}
 F_\epsilon (\thetav,\etav) &= (2\pi)^{-2} \iint dx\,dy\, h_1(\epsilon x) h_2(\epsilon y) e^{i \varphi_A(\thetav)Qx + i (p(\etav)-x) \cdot y}f(\thetav)g(\etav)
\\
 &= \int dx\,  h_1(\epsilon x) \frac{1}{2\pi\epsilon} \tilde h_2\big(\epsilon^{-1}(p(\etav)-x)\big) \, 
   e^{i\varphi_A(\thetav)Qx} f(\thetav)g(\etav).
\end{aligned}
\end{equation}
Here $\tilde h_2$ restricts the integral to a compact set, and for sufficiently small $\epsilon$, we can then replace $h_1(\epsilon x)$ with $1$. Further, $\frac{1}{2\pi\epsilon} \tilde h_2(\epsilon^{-1} \cdotarg)$ is a delta sequence, so that as $\epsilon \to 0$,
\begin{equation}
 F_\epsilon (\thetav,\etav) \to e^{i\varphi_A(\thetav) Q p(\etav)} f(\thetav)g(\etav)\quad \text{in } \dcal(\rbb^{m+n}).
\end{equation}
Inserted into \eqref{eq:crlimit}, this gives \eqref{eq:QCvector}.
\end{proof}

Using the previous lemma, we can deduce a crucial relation regarding the product of two deformed homogeneous distributions.

\begin{lemma}\label{lemma:homogdeformprod}
If $A,B$ are two homogeneous $\fcal^\infty$-valued distributions, then, in the sense of distributional kernels,
\begin{equation}\label{eq:homogdeformprod}
\tau_{Q}(A(\thetav)) \tau_{Q}(B(\etav)) = e^{i\varphi_A(\thetav)Q\varphi_B(\etav)}\tau_{Q}(A(\thetav)B(\etav)).
\end{equation}
\end{lemma}

\begin{proof}
First, we remark that the support of the distribution $\hscalar{\lvector{}{\thetav'}}{ A(\thetav)B(\etav) \rvector{}{\etav'}}$ is concentrated on the hypersurface $p(\thetav')-p(\etav') = \varphi_A(\thetav)+\varphi_B(\etav)$. This is seen by computing, for $x \in \rbb^2$,
\begin{equation}\label{eq:trans1}
\begin{aligned}
  \hscalar{\lvector{}{\thetav'}}{ A(\thetav)B(\etav) \rvector{}{\etav'}} 
 &= \hscalar{U(x) \lvector{}{\thetav'}}{ U(x) A(\thetav)B(\etav)  U(x)^\ast\,U(x) \rvector{}{\etav'}} 
\\
 &= e^{i(-p(\thetav')+p(\etav')+\varphi_A(\thetav)+\varphi_B(\etav))\cdot x}\hscalar{\lvector{}{\thetav'}}{ A(\thetav)B(\etav) \rvector{}{\etav'}},
\end{aligned}
\end{equation}
where homogeneity of $A,B$ and covariance properties of $\lvector{}{\cdotarg}, \rvector{}{\cdotarg}$ have been used.
But \eqref{eq:trans1} can hold for all $x$ only if the support of the distribution is contained in the  surface mentioned.

Now we compute by applying Lemma~\ref{lemma:homogdeformvec} twice,
\begin{equation}\label{eq:qab1}
\begin{aligned}
 \hscalar{\lvector{}{\thetav'}}{ \tau_Q(A(\thetav)) \tau_Q(B(\etav)) \rvector{}{\etav'}}
&= e^{i \varphi_A(\thetav) Q p(\thetav')} e^{i \varphi_B(\etav) Q p(\etav')}
 \hscalar{\lvector{}{\thetav'}}{ A(\thetav)B(\etav) \rvector{}{\etav'}}
\\
&= e^{i \varphi_A(\thetav) Q \varphi_B(\etav)} e^{i p(\thetav') Q p(\etav')}
 \hscalar{\lvector{}{\thetav'}}{ A(\thetav)B(\etav) \rvector{}{\etav'}},
\end{aligned}
\end{equation}
where we used skew-symmetry of $Q$ and the support property of the distribution as remarked above. Likewise, we obtain
\begin{equation}\label{eq:qab2}
\begin{aligned}
 \hscalar{\lvector{}{\thetav'}}{ \tau_Q(A(\thetav) B(\etav)) \rvector{}{\etav'}}
&= e^{i (\varphi_A(\thetav) +\varphi_B(\etav)) Q p(\etav')} 
 \hscalar{\lvector{}{\thetav'}}{ A(\thetav)B(\etav) \rvector{}{\etav'}}
\\
&= e^{i p(\thetav') Q p(\etav')}
 \hscalar{\lvector{}{\thetav'}}{ A(\thetav)B(\etav) \rvector{}{\etav'}}.
\end{aligned}
\end{equation}
Eqs.~\eqref{eq:qab1} and \eqref{eq:qab2} together imply the result.
\end{proof}

At this stage, we can identity the deformed theory with an integrable model. To that end we set $\zd(\theta)=\tau_Q(\ad(\theta))$, $z(\eta)=\tau_Q(a(\eta))$. Applying Lemma~\ref{lemma:homogdeformprod} twice, we find that these $z,\zd$ fulfill
\begin{equation}\label{eq:zzrfromq}
\begin{aligned}
\zd(\theta)\zd(\theta') &= e^{2ip(\theta)Qp(\theta')}\zd(\theta')\zd(\theta),\\
z(\eta)z(\eta') &= e^{2ip(\eta)Qp(\eta')}z(\eta')z(\eta),\\
z(\eta)\zd(\theta) &= e^{2ip(\theta)Qp(\eta)}\zd(\theta)z(\eta) + \delta(\theta-\eta) \idop.
\end{aligned}
\end{equation}
Now there is only a one-parameter family of $2\times 2$ matrices which are skew symmetric with respect to the Minkowski scalar product; we write them as
\begin{equation}
 Q = -\frac{a}{2\mu^2} \begin{pmatrix} 0 & 1 \\ 1 & 0 \end{pmatrix}
\end{equation}
with a dimensionless real constant $a$. With this, the equations \eqref{eq:zzrfromq} are just the Zamolodchikov relations with the scattering function
\begin{equation}\label{eq:QSform}
   S(\theta) = e^{i a \sinh \theta}.
\end{equation}
We can then unitarily identify our ``free'' Hilbert space $\hcal$ with the $S$-symmetric Fock space over $\hcal_1$ as introduced in Sec.~\ref{sec:hilbertspace}, mapping the Zamolodchikov operators $z,\zd$ to their counterparts defined there, while preserving space-time translations and boosts. For details of this isomorphism, see \cite[Lemma~5.7]{Lechner:2008}.

We now proceed to define the $Q$-commutator, as announced.

\begin{definition}
For $A,B \in \ccal^\infty$, the $Q$-commutator is
\begin{equation}\label{Qcom}
[A,B]_{Q}:=AB -\tau_{2Q}\Big(\tau_{-2Q}(B)\tau_{-2Q}(A)\Big).
\end{equation}
We use the same definition if $A,B \in \qcal^\infty$ and at least \emph{one} of them is in $\fcal^\infty$.
\end{definition}

For homogeneous distributions $A(\thetav),B(\etav)$, we can compute an explicit expression for the $Q$-commutator:
from Lemma~\ref{lemma:homogdeformprod} and Prop.~\ref{proposition:tauq}\ref{it:warpedcompose} we obtain
\begin{equation}\label{eq:Qcommhomog}
[A(\thetav),B(\etav)]_{Q}=A(\thetav)B(\etav)-e^{2i\varphi_{A}(\thetav)Q\varphi_{B}(\etav)}B(\etav)A(\thetav).
\end{equation}
In particular, the $Q$-commutator expression is again homogeneous.

We note that the $Q$-commutator is bilinear and fulfills the following ``deformed'' versions of the standard properties of a commutator. We formulate them for homogeneous distributions only. For general elements of $\qf^\infty$ or $\ccal^\infty$, we could obtain similar relations by decomposing them into homogeneous distributions, either by a spectral decomposition in the sense of Arveson \cite{Arveson:1974} with respect to the action of the translation group, or indeed by using the operator expansion of Thm.~\ref{theorem:arakiexpansion}.

\begin{proposition}
For homogeneous $\fcal^\infty$-valued distributions with kernels $A(\thetav), B(\etav), C(\xiv)$, the $Q$-commutator satisfies

\begin{enumerate}
\renewcommand{\theenumi}{(\roman{enumi})}
\renewcommand{\labelenumi}{\theenumi}
\item\label{it:qcommanti}
 anticommutativity:
\begin{equation}
[A(\thetav),B(\etav)]_{Q}=-e^{2i\varphi_{A}(\thetav)Q\varphi_{B}(\etav)}[B(\etav),A(\thetav)]_{Q};
\end{equation}

\item \label{it:qcommleibniz}
Leibniz rule:
\begin{equation}\label{eq:deformedleibniz}
[A(\thetav),B(\etav)C(\xiv)]_{Q}=[A(\thetav),B(\etav)]_{Q}C(\xiv)+e^{2i\varphi_{A}(\thetav)Q\varphi_{B}(\etav)}B(\etav)[A(\thetav),C(\xiv)]_{Q};
\end{equation}

\item\label{it:qcommjacobi}
 Jacobi identity:
\begin{equation}
e^{-2i\varphi_{A}(\thetav)Q\varphi_{C}(\xiv)}[A(\thetav),[B(\etav),C(\xiv)]_{Q}]_{Q}+ \mathrm{ cyclic\, permutations } =0.\label{Qjacobi}
\end{equation}
\end{enumerate}
\end{proposition}

All three relations can be obtained by repeated application of Eq.~\eqref{eq:Qcommhomog}; the computation is  straightforward, and we omit it here.

As another direct consequence of Eq.~\eqref{eq:Qcommhomog}, we can rewrite the Zamolodchikov relations \eqref{eq:zzrfromq} in terms of $Q$-commutators as follows:
\begin{equation}
 [z^{\dagger}(\theta),z^{\dagger}(\theta')]_{Q}=0,
\quad
  [z(\eta),z(\eta')]_{Q}=0,
\quad
  [z(\eta),z^{\dagger}(\theta)]_{Q}=\delta(\theta-\eta) \idop.
\end{equation}
That is, the Zamolodchikov operators fulfill CCR-like relations with respect to the $Q$-commutator; this stresses the analogy to the graded commutator in the CAR case.
Moreover, we can obtain again from Eq.~\eqref{eq:Qcommhomog},
\begin{equation}
[z(\xi),z^{\dagger m}(\thetav)z^{n}(\etav)]_{Q}=z(\xi)z^{\dagger m}(\thetav)z^{n}(\etav)
  -e^{2i(p(\thetav)-p(\etav))Qp(\xi)}z^{\dagger m}(\thetav)z^{n}(\etav)z(\xi),
\end{equation}
which implies by repeated application of the Zamolodchikov relations \eqref{eq:zzrfromq},
\begin{multline}
[z(\xi),z^{\dagger m}(\thetav)z^{n}(\etav)]_{Q}=\sum_{j=1}^{m}\Big(\prod_{l=1}^{j-1}e^{2ip(\theta_{l})Qp(\xi)}\Big)\delta(\theta_{j}-\xi)z^{\dagger}(\theta_{1})\ldots \widehat{z^{\dagger}(\theta_{j})}\ldots z^{\dagger}(\theta_{m})z^{n}(\etav)\\
=m\operatorname{Sym}_{S^{-1},\thetav}\Big( \delta(\xi-\theta_{1})z^{\dagger m-1}(\theta_{2},\ldots,\theta_{m})z^{n}(\etav)\Big).\label{QcommAz}
\end{multline}
Similarly, we have
\begin{equation}
[z^{\dagger m}(\thetav)z^{n}(\etav),z^{\dagger}(\xi)]_{Q}
=n \operatorname{Sym}_{S^{-1},\etav}\Big( \delta(\xi-\eta_{n})z^{\dagger m}(\thetav)z^{n-1}(\eta_{1},\ldots ,\eta_{n-1})\Big).\label{QcommAzdag}
\end{equation}
This finally enables us to prove the proposed form of the expansion coefficients.

\begin{theorem}
Let $S$ be of the form \eqref{eq:QSform}. The coefficients $\cme{m,n}{A}$, where $A\in \qf^\infty$, can be expressed as
\begin{equation}\label{Snested}
\cme{m,n}{A}(\thetav,\etav)= \bighscalar{\Omega}{ [z(\theta_m), \ldots [z(\theta_1),[\ldots[A,\zd(\eta_n)]_Q\ldots ,\zd(\eta_1)]_Q \ldots ]_Q \;\Omega  }.
\end{equation}
\end{theorem}
\begin{proof}
We first remark that with $A$, also its expansion terms $z^{\dagger m} z^{n}(\cme{m,n}{A})$ are elements of $\qf^\infty$. Namely, using Prop.~\ref{proposition:fmnpoincare}, we find for the derivatives $\partial^\kappa$ with a multi-index $\kappa$,
\begin{equation}
   \partial^\kappa z^{\dagger m} z^{n}(\cmelong{m,n}{A}) = z^{\dagger m} z^{n}(\cmelong{m,n}{\partial^\kappa A});
\end{equation}
due to Prop.~\ref{proposition:zzdcrossnorm} and \ref{proposition:fmnbound}, these have finite norms when restricted to spaces of fixed particle number.

By Thm.~\ref{theorem:arakiexpansion}, it therefore suffices to prove the statement \eqref{Snested} for $A=z^{\dagger m'} z^{n'}(f)$.
Now for this particular $A$, or rather for its kernel $A(\thetav',\etav')=z^{\dagger m'}(\thetav') z^{n'}(\etav')$, the nested $Q$-commutator in \eqref{Snested} gives by repeated application of Eqs.~\eqref{QcommAz} and \eqref{QcommAzdag}:
\begin{multline}\label{applcomm}
 \bighscalar{\Omega}{ [z(\theta_m), \ldots [z(\theta_1),[\ldots[A,\zd(\eta_n)]_Q\ldots ,\zd(\eta_1)]_Q \ldots ]_Q \;\Omega  }
=m!n!\operatorname{Sym}_{S^{-1},\thetav'}\operatorname{Sym}_{S^{-1},\etav'}
\\
\Big( \prod_{j=1}^m \delta(\theta_j-\theta_{j}')
\prod_{k=0}^{n-1} \delta(\eta_{n-k}-\eta_{n'-k}')\;
z^{\dagger m'-m}(\theta'_{m+1},\ldots,\theta_{m'}') 
z^{n'-n}(\eta'_{1},\ldots,\eta_{n'-n}') 
\Big)
\end{multline}
if $m'\geq m$, $n'\geq n$, and the right hand side vanishes otherwise. Now if $m'>m$ or $n'>n$, the vacuum expectation value of the right hand side of \eqref{applcomm} vanishes. Therefore, we find:
\begin{equation}
 \begin{aligned}
 \bighscalar{\Omega}{ [z(\theta_m), \ldots &[z(\theta_1),[\ldots[A,\zd(\eta_n)]_Q\ldots ,\zd(\eta_1)]_Q \ldots ]_Q \;\Omega  } 
\\
&=m!n!\delta_{m,m'}\delta_{n,n'}\operatorname{Sym}_{S^{-1},\thetav'}\operatorname{Sym}_{S^{-1},\etav'}\Big( \delta^{m}(\thetav-\thetav') \delta^{n}(\etav-\etav')\Big)\\
&= m!n!\delta_{m,m'}\delta_{n,n'}\operatorname{Sym}_{S,\thetav}\delta^{m}(\thetav-\thetav')\operatorname{Sym}_{S,\etav}\delta^{n}(\etav-\etav').
 \end{aligned}
\end{equation}
We have used \eqref{symstheta} here. This matches the left hand side of \eqref{Snested} because of Lemma~\ref{lemma:fmnbasis}.
\end{proof}

The theorem shows that the expansion coefficients can be expressed in terms of deformed commutators if the scattering function is of the form \eqref{eq:QSform}. It would be interesting to find similar structures for general $S$. On a formal level, this should in fact be possible: One could make use of the more general deformation scheme in \cite{Lechner:2011} to derive a suitable ``$S$-commutator''. Alternatively, and more directly, one could use the operator expansion itself to define the deformed commutator, imposing the relations $[z(\eta),\zd(\theta)]_S = \delta(\theta-\eta)\idop$, etc., and extending them by a deformed Leibniz rule as in \eqref{eq:deformedleibniz}. However, while such a definition might make sense on a formal level, its functional analytic properties (e.g., whether the $S$-commutator of two bounded operators would be bounded) remain unclear at this time.

\section{Conclusions and outlook}\label{sec:conclusion}

In this paper, we have established the operator expansion \eqref{eq:expansionz} in integrable models in a precise sense, defined on a space of quadratic forms. We found explicit (if intricate) expressions for the expansion coefficients $\cme{m,n}{A}$ as linear functionals of the quadratic form $A$. In some models, these can be rewritten in terms of nested ``deformed commutators'', in generalization of the free field situation.

We remark that it is also possible to visualize the combinatorics of the expansion coefficients using a diagrammatic representation: Given a contraction $C$, the pairs of contracted indices would correspond to internal lines, and the non-contracted indices to external lines. Each internal line would yield a factor $-\delta(\cdotarg)$, combining to the factor $(-1)^{|C|} \delta_C$ in \eqref{eq:fmndef}, and each crossing of lines would indicate an $S$-factor, combining to the factor $S_C$. Since this notation is not directly relevant to the analysis at hand, we do not elaborate it further here.

As a next step \cite{BostelmannCadamuro:characterization-wip}, we are planning to characterize the localization of $A$ in bounded regions in terms of analyticity properties of its expansion coefficients $\cme{m,n}{A}$, making precise what was suggested in \cite{SchroerWiesbrock:2000-1}. For this purpose, the behavior of $\cme{m,n}{A}$ under space-time symmetries, in particular reflections, as established in Sec.~\ref{sec:fsymtrans} will play a crucial role. 

We have restricted our attention to theories with a rather simple particle spectrum, namely, consisting of a single species of uncharged scalar particle. This was mostly to avoid formal complications. The overall framework we use can be generalized to a richer particle spectrum \cite{LechnerSchuetzenhofer:2012}, where the scattering function $S$ is then replaced with a matrix-valued function. The structural results for the operator expansion should essentially be the same in the more general case.

It should also be noted that the operator expansion in Eq.~\eqref{eq:expansionz}, as written, is applicable only to situations where the scattering function $S$ has no poles in the physical strip. This may not be apparent, since Thm.~\ref{theorem:arakiexpansion} does not require any analytic continuation of $S$ at all. Yet in models with poles in the physical strip (such as sine-Gordon), the results of \cite{Lechner:2008} about local operators do not apply, and the Hilbert space as defined in Sec.~\ref{sec:hilbertspace} is likely too small to accomodate local observables. It would need to be extended in order to include extra states corresponding to the poles (``bound states''). Nevertheless, we expect that a suitably generalized version of Thm.~\ref{theorem:arakiexpansion} will hold on the extended Hilbert space.

Since our methods fit into the context of deformations of quantum field theories, specifically, of warped convolutions, and since these give a generalization of integrable models to higher space-time dimensions, it is natural to ask whether the operator expansion can be generalized to higher space-time dimensions as well. Using methods as established in Sec.~\ref{sec:warped}, we do indeed think that this is the case. As in the 1+1 dimensional situation, the ``basis'' $z^{\dagger m} z^n$ of the expansion would depend on the choice of the deformation matrix $Q$; but in higher dimensions, there is a much larger choice for $Q$, and in more geometrical terms, the expansion would depend on the choice of a wedge region in Minkowski space. Since however the literature suggests \cite{BuchholzSummers:2007,BuchholzSummersLechner:2011} that these higher-dimensional models contain few, if any, observables localized in bounded regions, this falls somewhat outside our proposed line of investigation.

Besides the possibility of characterizing local operators, to be discussed elsewhere \cite{BostelmannCadamuro:characterization-wip}, we think that the present operator expansion may be helpful in resolving further open problems in the nonperturbative treatment of integrable quantum field theories. Like in free field theory, the expansion should be linked to the point field structure of the theory, allowing to identify pointlike fields and possibly their operator product expansions \cite{WilZim:products,Bos:product_expansions}. It might also shed light on the structure of the massless limit, or short-distance scaling limit, of integrable models, where the size of the local observable algebras is still an open problem \cite{BostelmannLechnerMorsella:2011}.

\section*{Acknowledgements}

The authors are indebted to K.-H.~Rehren for valuable suggestions and comments. We~would also like to thank M.~Bischoff, W.~Dybalski, C.~J.~Fewster, G.~Lechner and Y.~Tanimoto for helpful comments and discussions.

Work on this article was partially done during several stays at the Erwin Schr\"odinger Institute, Vienna, and the authors are grateful to ESI for hospitality and, in the case of H.~B., for financial support. The work of D.~C.~was supported by the German Research Foundation (Deutsche Forschungsgemeinschaft (DFG)) through 
the Institutional Strategy of the University of G\"ottingen. H.~B.~would like to thank the University of G\"ottingen for an invitation.


\bibliographystyle{abbrv} 
\bibliography{../../integrable}

\begin{thebibliography}{10}

\bibitem{Ara:lattice}
H.~Araki.
\newblock A lattice of von {N}eumann algebras associated with the quantum
  theory of a free {B}ose field.
\newblock {\em J.~Math. Phys.}, 4:1343--1362, 1963.

\bibitem{Arveson:1974}
W.~Arveson.
\newblock On groups of automorphisms of operator algebras.
\newblock {\em Journal of Functional Analysis}, 15(3):217--243, 1974.

\bibitem{BabujianFoersterKarowski:2006}
H.~M. Babujian, A.~Foerster, and M.~Karowski.
\newblock The form factor program: {A} review and new results -- the nested
  {SU(N)} off-shell {B}ethe ansatz.
\newblock {\em SIGMA}, 2:082, 2006.

\bibitem{BabujianKarowski:2004}
H.~M. Babujian and M.~Karowski.
\newblock {Towards the construction of Wightman functions of integrable quantum
  field theories}.
\newblock {\em Int. J. Mod. Phys.}, A19S2:34--49, 2004.

\bibitem{BorYng:positivity}
H.-J. Borchers and J.~Yngvason.
\newblock Positivity of {Wightman} functionals and the existence of local nets.
\newblock {\em Commun. Math. Phys.}, 127:607--615, 1990.

\bibitem{Bos:product_expansions}
H.~Bostelmann.
\newblock Operator product expansions as a consequence of phase space
  properties.
\newblock {\em J.~Math. Phys.}, 46:082304, 2005.

\bibitem{Bos:short_distance_structure}
H.~Bostelmann.
\newblock Phase space properties and the short distance structure in quantum
  field theory.
\newblock {\em J.~Math. Phys.}, 46:052301, 2005.

\bibitem{BostelmannCadamuro:characterization-wip}
H.~Bostelmann and D.~Cadamuro.
\newblock Characterization of local observables in integrable quantum field
  theories.
\newblock In preparation.

\bibitem{BostelmannCadamuro:examples-wip}
H.~Bostelmann and D.~Cadamuro.
\newblock Towards an explicit construction of local observables in integrable
  quantum field theories.
\newblock In preparation.

\bibitem{BostelmannDAntoniMorsella:2010}
H.~Bostelmann, C.~D'Antoni, and G.~Morsella.
\newblock On dilation symmetries arising from scaling limits.
\newblock {\em Commun. Math. Phys.}, 294:21--60, 2010.

\bibitem{BostelmannLechnerMorsella:2011}
H.~Bostelmann, G.~Lechner, and G.~Morsella.
\newblock Scaling limits of integrable quantum field theories.
\newblock {\em Rev.~Math.~Phys.}, 23:1115--1156, 2011.

\bibitem{BuchholzSummersLechner:2011}
D.~Buchholz, G.~Lechner, and S.~J. Summers.
\newblock Warped convolutions, {R}ieffel deformations and the construction of
  quantum field theories.
\newblock {\em Commun. Math. Phys.}, 304:95--123, 2011.

\bibitem{BucPor:phase_space}
D.~Buchholz and M.~Porrmann.
\newblock How small is the phase space in quantum field theory?
\newblock {\em Ann. Inst. H.~Poincar{\'e}}, 52:237--257, 1990.

\bibitem{BuchholzSummers:2006}
D.~Buchholz and S.~J. Summers.
\newblock Scattering in relativistic quantum field theory: Fundamental concepts
  and tools.
\newblock In J.-P. Fran\c{c}oise, G.~L. Naber, and T.~S. Tsun, editors, {\em
  Encyclopedia of Mathematical Physics}, pages 456--465. Academic Press,
  Oxford, 2006.

\bibitem{BuchholzSummers:2007}
D.~Buchholz and S.~J. Summers.
\newblock String- and brane-localized causal fields in a strongly nonlocal
  model.
\newblock {\em J. Phys. A: Math. Gen.}, 40(9):2147, 2007.

\bibitem{BuchholzSummers:2008}
D.~Buchholz and S.~J. Summers.
\newblock {Warped Convolutions: a Novel Tool in the Construction of Quantum
  Field Theories}.
\newblock In E.~{Seiler} and K.~{Sibold}, editors, {\em Quantum Field Theory
  and Beyond. Essays in Honor of Wolfhart Zimmermann}, pages 107--121, Oct.
  2008.

\bibitem{Cadamuro:2012}
D.~Cadamuro.
\newblock {\em A Characterization Theorem for Local Operators in Factorizing
  Scattering Models}.
\newblock {Ph.D.}~thesis, Universit{\"a}t G{\"o}ttingen, 2012.
\newblock Available electronically as arXiv:1211.3583.

\bibitem{Doplicher:2010}
S.~Doplicher.
\newblock Spin and statistics and first principles.
\newblock {\em Foundations of Physics}, 40:719--732, 2010.

\bibitem{FreHer:pointlike_fields}
K.~Fredenhagen and J.~Hertel.
\newblock Local algebras of observables and pointlike localized fields.
\newblock {\em Commun. Math. Phys.}, 80:555--561, 1981.

\bibitem{GliJaf:quantum_physics}
J.~Glimm and A.~Jaffe.
\newblock {\em Quantum Physics -- A functional integral point of view}.
\newblock Springer, New York, 2nd edition, 1987.

\bibitem{GrosseLechner:2007}
H.~Grosse and G.~Lechner.
\newblock Wedge-local quantum fields and noncommutative {M}inkowski space.
\newblock {\em Journal of High Energy Physics}, 2007(11):012, 2007.

\bibitem{Haa:LQP}
R.~Haag.
\newblock {\em Local Quantum Physics}.
\newblock Springer, Berlin, 2nd edition, 1996.

\bibitem{Jaffe:1967}
A.~M. Jaffe.
\newblock High-energy behavior in quantum field theory. {I.}~{S}trictly
  localizable fields.
\newblock {\em Phys. Rev.}, 158:1454--1461, Jun 1967.

\bibitem{KadRin:algebras1}
R.~V. Kadison and J.~R. Ringrose.
\newblock {\em Fundamentals of the Theory of Operator Algebras}, volume I:
  Elementary Theory.
\newblock Academic Press, 1997.

\bibitem{Lashkevich:1994}
M.~Y. Lashkevich.
\newblock Sectors of mutually local fields in integrable models of quantum
  field theory.
\newblock June 1994.
\newblock arXiv:hep-th/9406118.

\bibitem{Lechner:2003}
G.~Lechner.
\newblock {Polarization-free quantum fields and interaction}.
\newblock {\em Lett. Math. Phys.}, 64:137--154, 2003.

\bibitem{Lechner:2006}
G.~Lechner.
\newblock {\em {On the construction of quantum field theories with factorizing
  S-matrices}}.
\newblock PhD thesis, University of G{\"o}ttingen, 2006.

\bibitem{Lechner:2008}
G.~Lechner.
\newblock Construction of quantum field theories with factorizing {S}-matrices.
\newblock {\em Commun. Math. Phys.}, 277:821--860, 2008.

\bibitem{Lechner:2011}
G.~Lechner.
\newblock Deformations of quantum field theories and integrable models.
\newblock {\em Commun. Math. Phys.}, 312:265--302, 2012.

\bibitem{LechnerSchuetzenhofer:2012}
G.~Lechner and C.~Sch{\"u}tzenhofer.
\newblock Towards an operator-algebraic construction of integrable global gauge
  theories, Aug. 2012.
\newblock Preprint arXiv:1208.2366.

\bibitem{LechnerWaldmann:2011}
G.~Lechner and S.~Waldmann.
\newblock Strict deformation quantization of locally convex algebras and
  modules, 2011.
\newblock Preprint arXiv:1109.5950v1.

\bibitem{RehWol:fields_localized}
J.~Rehberg and M.~Wollenberg.
\newblock Quantum fields as pointlike localized objects.
\newblock {\em Math. Nachr.}, 125:259--274, 1986.

\bibitem{Rieffel:1993}
M.~Rieffel.
\newblock {\em Deformation quantization for actions of $\mathbb{R}^d$}, volume
  506 of {\em Memoirs of the AMS}.
\newblock American Mathematical Society, 1993.

\bibitem{Schroer:1999}
B.~Schroer.
\newblock Modular wedge localization and the d=1+1 formfactor program.
\newblock {\em Annals Phys.}, 275:190--223, 1999.

\bibitem{SchroerWiesbrock:2000-1}
B.~Schroer and H.~W. Wiesbrock.
\newblock {Modular constructions of quantum field theories with interactions}.
\newblock {\em Rev. Math. Phys.}, 12:301--326, 2000.

\bibitem{Smirnov:1992}
F.~A. Smirnov.
\newblock {\em {Form Factors in Completely Integrable Models of Quantum Field
  Theory}}.
\newblock World Scientific, Singapore, 1992.

\bibitem{StrWig:PCT}
R.~F. Streater and A.~S. Wightman.
\newblock {\em PCT, Spin and Statistics, and All That}.
\newblock Benjamin, New York, 1964.

\bibitem{WilZim:products}
K.~G. Wilson and W.~Zimmermann.
\newblock Operator product expansions and composite field operators in the
  general framework of quantum field theory.
\newblock {\em Commun. Math. Phys.}, 24:87--106, 1972.

\end{thebibliography}

\end{document}